\documentclass[11pt,english]{article}
\usepackage[T1]{fontenc}
\usepackage[latin9]{inputenc}
\usepackage{amsthm}
\usepackage{amsmath}
\usepackage{amssymb}

\usepackage{algpseudocode}
\usepackage{algorithm}

\usepackage{graphicx}

\usepackage{subcaption}

\usepackage[top=1in, bottom=1in, left=1in, right=1in]{geometry}

\makeatletter
\theoremstyle{plain}
\newtheorem{thm}{\protect\theoremname}[section]
  \theoremstyle{definition}
  \newtheorem{defn}[thm]{\protect\definitionname}
  \theoremstyle{remark}
  \newtheorem{claim}[thm]{\protect\claimname}
  \theoremstyle{plain}
  \newtheorem{lem}[thm]{\protect\lemmaname}
  \theoremstyle{plain}
  
  \newtheorem{fact}[thm]{Fact}

\usepackage{fullpage}
\usepackage{color}

\makeatother

\usepackage{babel}
  \providecommand{\claimname}{Claim}
  \providecommand{\definitionname}{Definition}
  \providecommand{\lemmaname}{Lemma}
\providecommand{\theoremname}{Theorem}

\mathchardef\mhyphen="2D

\newcommand{\Next}{\textup{Next}}
\newcommand{\cmatch}{c_{\mathrm{match}}}
\newcommand{\cdelx}{c_{\mathrm{del}\mhyphen\mathrm{x}}}
\newcommand{\cdely}{c_{\mathrm{del}\mhyphen\mathrm{y}}}
\newcommand{\csubst}{c_{\mathrm{subst}}}
\newcommand{\OVH}{{OVH}}
\newcommand{\UOVH}{{UOVH}}
\newcommand{\x}{{\textsc{x}}}
\newcommand{\y}{{\textsc{y}}}
\newcommand{\strc}{{\cal S}}
\newcommand{\aligngad}{alignment}
\newcommand{\Aligngad}{Alignment}
\newcommand{\EDIT}{\textup{Edit}}
\newcommand{\Edit}{\EDIT}
\newcommand{\EDITallc}{\ensuremath{\EDIT(\cdelx,\cdely,\cmatch,\csubst)}}
\newcommand{\EDITc}{\ensuremath{\EDIT(\csubst)}}
\newcommand{\DTW}{\textup{DTW}}
\newcommand{\dDTW}{\delta_\textup{DTW}}
\newcommand{\dLCS}{\delta_\textup{LCS}}

\newcommand{\LPS}{\textup{LPS}}
\newcommand{\LTS}{\textup{LTS}}

\newcommand{\dEDITc}{\delta_{\EDITc}}
\newcommand{\dEDIT}{\delta_{\EDIT}}
\newcommand{\eps}{\ensuremath{\varepsilon}}

\newcommand{\Val}[1]{\delta(#1)}
\newcommand{\rev}{\mathrm{rev}}

\newcommand{\claref}[1]{Claim~\ref{cla:#1}}
\renewcommand{\algref}[1]{Algorithm~\ref{alg:#1}}

\newcommand{\figref}[1]{Figure~\ref{fig:#1}}

\newcommand{\defref}[1]{Definition~\ref{def:#1}}
\newcommand{\facref}[1]{Fact~\ref{fac:#1}}

\newcommand{\thmref}[1]{Theorem~\ref{thm:#1}}
\newcommand{\thmrefs}[2]{Theorems~\ref{thm:#1} and~\ref{thm:#2}}
\newcommand{\thmrefss}[3]{Theorems~\ref{thm:#1}, \ref{thm:#2}, and~\ref{thm:#3}}

\newcommand{\lemref}[1]{Lemma~\ref{lem:#1}}

\newcommand{\secref}[1]{Section~\ref{sec:#1}}

\newcommand{\eq}[1]{equation~\eqref{eq:#1}}
\newcommand{\ineq}[1]{inequality~\eqref{eq:#1}}

\begin{document}
\global\long\def\SC{\mathbf{\Pi}}

\global\long\def\CG{\mathrm{CG}}

\global\long\def\VG{\mathrm{VG}}
\global\long\def\NVG{\mathrm{NVG}}

\global\long\def\CandA{\mathrm{GA}}

\global\long\def\zeroes{\mathrm{zeroes}}

\global\long\def\guard{\mathrm{G}}

\global\long\def\dist{\mathrm{dist}}

\global\long\def\Vinner{V_{\mathrm{in}}}

\global\long\def\Vouter{V_{\mathrm{out}}}

\global\long\def\Xz{X^{\mathbf{0}}}

\global\long\def\Xo{X^{\mathbf{1}}}

\global\long\def\Yz{Y^{\mathbf{0}}}

\global\long\def\Yo{Y^{\mathbf{1}}}

\global\long\def\Cz{C^{\mathbf{0}}}

\global\long\def\Co{C^{\mathbf{1}}}

\global\long\def\Gz{G^{\mathbf{0}}}

\global\long\def\Go{G^{\mathbf{1}}}

\global\long\def\inputs{{\cal I}}

\global\long\def\algn{{\cal A}}

\global\long\def\type{\mathrm{type}}

\global\long\def\occ{\mathrm{occ}}

\global\long\def\poly{\mathrm{poly}}

\global\long\def\bigOh{{\cal O}}

\global\long\def\zleft{\mathbf{0}_{\x}}

\global\long\def\zright{\mathbf{0}_{\y}}

\global\long\def\oleft{\mathbf{1}_{\x}}

\global\long\def\oright{\mathbf{1}_{\y}}

\global\long\def\dtw{\DTW}
\global\long\def\cost{d}

\global\long\def\LCS{\textup{LCS}}
\global\long\def\OV{\textup{OV}}
\global\long\def\SETH{\textup{SETH}}
\global\long\def\kSAT{\textup{$k$-SAT}}

\global\long\def\Oh{\bigOh}

\newcommand{\sugg}[2]{{\color{green} #1} {\color{red}\tiny  #2}}

\title{Quadratic Conditional Lower Bounds for String Problems \\ and Dynamic Time Warping}
\author{Karl Bringmann\thanks{Institute of Theoretical Computer Science, ETH Zurich, Switzerland, \texttt{karlb@inf.ethz.ch}} \and Marvin K\"unnemann\thanks{Max Planck Institute for Informatics, Saarbr\"ucken, Germany, \texttt{marvin@mpi-inf.mpg.de}}}
\maketitle

\medskip

\begin{abstract}
  
  Classic similarity measures of strings are longest common subsequence and Levenshtein distance (i.e., the classic edit distance). A classic similarity measure of curves is dynamic time warping.
  These measures can be computed by simple $\Oh(n^2)$ dynamic programming algorithms, and despite much effort no algorithms with significantly better running time are known. 
  
  We prove that, even restricted to binary strings or one-dimensional curves, respectively, these measures do not have strongly subquadratic time algorithms, i.e., no algorithms with running time $\Oh(n^{2-\eps})$ for any $\eps > 0$, unless the Strong Exponential Time Hypothesis fails. 
  We generalize the result to edit distance for arbitrary fixed costs of the four operations (deletion in one of the two strings, matching, substitution), by identifying trivial cases that can be solved in constant time, and proving quadratic-time hardness on binary strings for all other cost choices. 
  This improves and generalizes the known hardness result for Levenshtein distance [Backurs, Indyk STOC'15] by the restriction to binary strings and the generalization to arbitrary costs, and
  adds important problems to a recent line of research showing conditional lower bounds for a growing number of quadratic time problems.

  As our main technical contribution, we introduce a framework 
  for proving quadratic-time hardness of similarity measures. To apply the framework it suffices to construct a single gadget, which encapsulates all the expressive power necessary to emulate a reduction from satisfiability. 
  
  Finally, we prove quadratic-time hardness for longest palindromic subsequence and longest tandem subsequence via reductions from longest common subsequence, showing that conditional lower bounds based on the Strong Exponential Time Hypothesis also apply to string problems that are not necessarily similarity measures.
\end{abstract}

\section{Introduction}

For many classic polynomial time problems the worst-case running time is stagnant for decades, e.g., a classic algorithm solves the problem in time $\tilde \Oh(n^2)$, up to logarithmic factors, but it is unknown whether any faster algorithms exist. 
For these problems we would like to explain why it is hard to find faster algorithms. One type of explanation is a \emph{conditional lower bound}. Here we assume that some problem $P$ has no algorithms faster than a long-standing time barrier and prove resulting lower bounds for other problems, via reductions from $P$. 
The most prominent such approach is 3SUM-hardness, which dates back to 1995~\cite{GajentaanO95}: Assuming that 3SUM has no (strongly) subquadratic algorithms, many lower bounds have been shown, especially for problems in computational geometry. However, for many other problems it seems to be impossible to find a reduction from 3SUM. 

In the last years, new assumptions emerged that allow to prove conditional lower bounds for problems where 3SUM-hardness does not seem to apply. The prime example is the Strong Exponential Time Hypothesis (\SETH), which was introduced by Impagliazzo and Paturi~\cite{ImpagliazzoP01} and asserts that satisfiability has no algorithms that are much faster than exhaustive search.

\paragraph{\textup{\textbf{Hypothesis}} \SETH:} \emph{For no $\eps > 0$, \kSAT\ can be solved in time $\Oh(2^{(1-\eps)N})$ for all $k \ge 3$.} 

\paragraph{}
Note that exhaustive search takes time $\Oh(2^N)$ and the best-known algorithms for \kSAT\ have a running time of the form $\Oh(2^{(1-c/k)N})$ for some constant $c>0$~\cite{paturi2005}. Thus, \SETH\ is a reasonable hypothesis and, due to lack of progress in the last decades, can be considered unlikely to fail. 

The idea to use \SETH\ to prove conditional lower bounds for polynomial time problems dates back to 2005~\cite{williams05}, but only in recent years more and more such conditional lower bounds have been proven, see, e.g.,~\cite{
abboud_quadratic_2015,AbboudVW14,AbboudWW14,Arthurs,
bringmann_walking_2014,PatrascuW10,RodittyVW13}. Two recent examples, that motivated this paper, are the conditional lower bounds for Fr\'echet distance~\cite{bringmann_walking_2014} and Levenshtein distance~\cite{Arthurs}.
Both problems are natural \emph{similarity measures} between two sequences (curves or strings, respectively). In this paper we study additional classic similarity measures between strings and curves. We propose a framework for proving lower bounds for such similarity measures. This allows us to prove quadratic-time hardness of the following problems. 

\paragraph{Edit Distance} 
Given two strings $x,y$ of length $n,m$ ($n \ge m$), 
we start in their first symbols at positions $(1,1)$ and traverse them up to their last symbols at positions $(n,m)$ using the following operations: If we are at positions $(i,j)$ we may (1) delete a symbol in $x$ (this costs $\cdelx$ and we advance to $(i+1,j)$), (2) delete a symbol in $y$ (this costs $\cdely$ and we advance to $(i,j+1)$), (3) match the current symbols, which is only possible if $x[i]=y[j]$ (this costs $\cmatch$ and we advance to $(i+1,j+1)$), or (4) substitute the current symbols, which is only possible if $x[i] \ne y[j]$ (this costs $\csubst$ and we advance to $(i+1,j+1)$). The minimum total cost of such a sequence of operations is called the edit distance of $x$ and $y$, and we denote the problem of computing the edit distance by \EDITallc. The Levenshtein distance (i.e., the classic edit distance) is $\EDIT(1,1,0,1)$. An important special case is the longest common subsequence (LCS) of two strings, which can be seen to be equivalent to $\EDIT(1,1,0,2)$. One obtains more variants for other cost choices, e.g., for aligning DNA sequences a classic choice is $\EDIT(2,2,-1,1)$~\cite{setubalmeidanis1997}.

Edit distance has a natural dynamic programming algorithm with running time $\Oh(nm)$, which is taught in many undergraduate algorithms courses. Since such string distance measures have many applications in bioinformatics and data comparison, Levenshtein distance and LCS are well-studied with a rich literature focussing on approximation algorithms (see, e.g.,~\cite{AndoniKO10}) and algorithms that perform well on special cases (see, e.g.,~\cite{hirschberg1977} and see~\cite{bergroth2000} for a survey). 
However, the best-known worst-case running time (of an exact algorithm) is $\Oh(n m / \log n + n)$~\cite{masek1980}, i.e., algorithms are stuck slightly below quadratic time. Even if we restrict the input to strings over a binary alphabet $\{0,1\}$ no better worst-case running time is known. In this paper we present a possible explanation for this situation by proving conditional lower bounds for edit distance on binary strings, thus improving and generalizing the known quadratic-time hardness for the Levenshtein distance on alphabet size~4~\cite{Arthurs}.

\paragraph{Dynamic Time Warping (\DTW)} 
Fix a metric space $(M,d)$. A sequence of points in $M$ is called a \emph{curve}. 
Consider two curves $x,y$ of length $n,m$ ($n \ge m$). We may \emph{traverse} $x$ and $y$ by starting in their first entries, in any time step advancing to the next entry in $x$ or $y$ or both, and ending in their last entries (see \secref{preliminaries} for details). The cost of such a traversal is the sum over all points in time of the distance between the current entries. The dynamic time warping distance of $x$ and $y$ is the minimal cost of any traversal. This similarity measure can, e.g., readily detect whether two given signals are equal up to time accelerations or decelerations. This property, among others, makes it a very useful measure in practice, with many applications in comparing temporal data such as video and audio, e.g., for speech recognition or music processing (see, e.g.,~\cite{sakoe1978}). The best-known worst-case running time is achieved by a simple dynamic programming algorithm that computes the \DTW\ distance of $x$ and $y$ in time $\Oh(nm)$. To break this apparent barrier in practice, many heuristics have been designed for this problem (see, e.g.,~\cite{salvador2007}).

An important special case that frequently arises in practice is \emph{dynamic time warping on one-dimensional curves}. Here the metric space is $M = \mathbb{R}$ and the distance measure is $d(a,b) := |a-b|$ for any $a,b \in \mathbb{R}$. Even for this important special case the best-known algorithm takes time $\Oh(n m)$. We provide a possible explanation for this situation by proving a conditional lower bound for \DTW\ on one-dimensional curves.

\subsection{Our Results}

\paragraph{Dynamic Time Warping}

As our first main result, we prove a conditional lower bound for \DTW. This shows that strongly subquadratic algorithms for \DTW\ can be considered unlikely to exist. Specifically, obtaining such algorithms is at least as hard as a breakthrough for satisfiability. 

\begin{thm} \label{thm:dtw}
  \DTW\ on one-dimensional curves taking values in $\{0,1,2,4,8\} \subseteq \mathbb{R}$ has no $\Oh(n^{2-\eps})$ algorithm for any $\eps>0$, unless \SETH\ fails.
\end{thm}

\paragraph{Edit Distance}

Our second main result is a classification of \EDITallc\ for all operation costs $\cdelx,\cdely,\cmatch,\csubst$: 
We identify trivial variants where the edit distance is independent of the input $x,y$, and only depends on $n,m$. In this case, it can be computed in constant time. For all remaining choices of the operation costs we prove quadratic-time hardness, even restricted to binary strings. This includes quadratic-time hardness of LCS and Levenshtein distance on binary strings. Compared to the known lower bound for Levenshtein distance~\cite{Arthurs}, our result decreases the alphabet size from 4 to 2 and adds hardness of a large class of problems including LCS.

\begin{thm} \label{thm:EDIT}
  \EDITallc\ can be solved in constant time if $\csubst = \cmatch$ or $\cdelx + \cdely \le \min\{\cmatch, \csubst\}$. Otherwise, \EDITallc\ on binary strings has no $\Oh(n^{2-\eps})$ algorithm for any $\eps>0$, unless \SETH\ fails.
\end{thm}

As first step of the hardness part of this theorem, for some $0 < \csubst' \le 2$ depending on $\cdelx,\cdely,\cmatch,\csubst$ we reduce $\EDIT(1,1,0,\csubst')$ to $\EDITallc$. This reduction is what fails for the trivial cases. Then we prove hardness of $\EDIT(1,1,0,\csubst')$ using a construction that is parameterized by $\csubst'$.

\paragraph{Unbalanced Inputs}
Our main results are most meaningful for inputs with $n \approx m$. It is conceivable that for unbalanced inputs, i.e., $m \ll n$, faster algorithms exist, say the running time of $\Oh(nm)$ could be reduced to $\tilde \Oh(n+m^2)$. For \DTW\ we show that such an improvement is unlikely, by proving that ``for any $m$'' no algorithm with running time $\Oh((nm)^{1-\eps})$ exists, assuming SETH. This is analogous to the situation for Fr\'echet distance~\cite{bringmann_walking_2014}.

\begin{thm} \label{thm:unbalanced}
  Unless \SETH\ fails, \DTW\ on one-dimensional curves taking values in $\{0,1,2,4,8\}$ has no  $\Oh((nm)^{1-\eps})$ algorithm for any $\eps>0$, and this even holds restricted to instances with $n^{\alpha-o(1)} \le m \le n^{\alpha+o(1)}$ for any $0 < \alpha < 1$.
\end{thm}

For edit distance, \thmref{EDIT} implies that there is no $\Oh(m^{2-\eps})$ algorithm for any $\eps > 0$ (in the worst case over all strings $x,y$ with $|x| \le n$ and $|y| \le m$ for any $n \ge m$). Our reduction from \SETH\ cannot result in unbalanced strings, and thus we are not able to prove better lower bounds than $\Oh(m^{2-\eps})$. This behaviour hints at the possibility of an $\tilde O(n+m^2)$ algorithm for edit distance - and indeed there is an algorithm for LCS from '77 due to Hirschberg~\cite{hirschberg1977} matching this time complexity. For completeness, we show that this algorithm can be generalized to edit distance.

\begin{thm} \label{thm:EDITalgo}
  \EDITallc\ has an $\tilde \Oh(n + m^2)$ algorithm.
\end{thm}

Thus, for unbalanced inputs \DTW\ and edit distance differ in their behaviour, but using \SETH\ we can readily explain this difference.

\paragraph{Reductions from Longest Common Subsequence}

Note that any near-linear time reduction from LCS to another problem $P$ transfers the quadratic-time lower bound of LCS to $P$. We think that this notion of \emph{LCS-hardness} could be used to prove lower bounds for many string problems (not only distance measures). To support this claim, we present two easy results in this direction. 

A \emph{palindromic subsequence} (also called symmetric subsequence) of a string $x$ of length $n$ is a subsequence $z$ that is the same as its reverse $\rev(z)$. Computing a longest palindromic subsequence is a popular exercise in undergraduate text books (e.g., \cite[Exercise 15-2]{CLRSthird}), since it can be easily solved by a reduction to LCS or adapting the dynamic programming solution of LCS, both resulting in an $\bigOh(n^2)$ algorithm. A \emph{tandem subsequence} of a string $x$ is a subsequence $z$ that can be written as the concatenation $z=yy$ of a string $y$ with itself. In contrast to longest palindromic subsequence, it is non-trivial to compute a longest tandem subsequence in time $\Oh(n^2)$~\cite{Kosowski2004}. We present reductions from LCS to both of these problems, which yields the following lower bounds.

\begin{thm} \label{thm:LPSLTS}
  On binary strings, longest palindromic subsequence and longest tandem subsequence 
  have no $\Oh(n^{2-\eps})$ algorithms for any $\eps>0$, unless \SETH\ fails.
\end{thm}

These results show that \SETH-based lower bounds via LCS are applicable to string problems that are not necessarily similarity measures.

\subsection{Technical Contribution}

We introduce a framework for proving \SETH-based lower bounds for similarity measures. It is based on a construction that we call \emph{\aligngad\ gadget}. Given instances $x_1,\ldots,x_n$ and $y_1,\ldots,y_m$, $m \le n$, an \aligngad\ gadget consists of two instances $x,y$ whose similarity $\delta(x,y)$ is closely related to $\sum_{(i,j) \in A} \delta(x_{i},y_j)$, where $A = \{(i_1,1),\ldots,(i_m,m)\}$ is the best-possible ordered alignment of the numbers in $[m]$ to $[n]$ (for details see \secref{framework}).
We prove a quadratic lower bound for any similarity measure admitting an \aligngad\ gadget. This proof is a simplified version of a construction in the known lower bound for Levenshtein distance~\cite{Arthurs}, which is also closely related to the lower bound for Fr\'echet distance~\cite{bringmann_walking_2014}. 

Working with our framework has two advantages: First, it unifies three constructions that are separate proof steps in other SETH-based lower bounds~\cite{Arthurs,bringmann_walking_2014}, thus reducing the amount of work necessary to prove SETH-based lower bounds. Second, it hides the reduction from satisfiability, providing a level of abstraction that allows to ignore the details of the satisfiability problem and instead focus on the details of the problem we reduce to. 
This makes it possible to tackle general problems such as \EDITallc, where the reduction depends on parameters of the problem, without resulting in an overly complex proof. 

We present \aligngad\ gadgets for edit distance and dynamic time warping. This part needs careful problem-specific constructions. In particular, we have to construct instances where the optimal sequence of edit distance operations has some exploitable structure, which is made difficult by the fact that we work over binary alphabet, so that in principle any two zeroes and any two ones can be matched.

\subsection{Related Work}

Independently of our work, similar lower bounds for \LCS\ and \DTW\ have been shown by Abboud et al.~\cite{abboud_quadratic_2015}. Let us briefly compare our approaches. 
Our main technical contribution is the \aligngad-framework, which allows us to give shorter hardness proofs. The proofs of Abboud et al.\ are longer, in particular since they are using the lower bound for Levenshtein distance~\cite{Arthurs}, while our proofs are self-contained. The main technical contribution of Abboud et al., apart from careful reductions, seems to be that they reduce from a novel problem that they call Most-Orthogonal Vectors. Regarding the problem \LCS, our hardness result is stronger, since we show hardness on binary strings, while Abboud et al.\ need alphabet size~7. Regarding \DTW, we prove hardness of different special cases, as we consider \DTW\ on one-dimensional curves over alphabets of size 5 (where the distance of two numbers is their absolute difference), while Abboud et al.\ consider \DTW\ on strings over alphabets of size 5 (where the distance of two symbols is 1 or 0, depending on whether they are equal or not). On top of these core results, Abboud et al.\ generalize their result for \LCS\ to $k$-\LCS, the longest common subsequence of $k$ strings. We classify the complexity of edit distance for arbitrary operation costs and prove hardness of additional string problems via reductions from LCS. 

\subsection{Organization}

In \secref{preliminaries} we fix notation and discuss alternative assumptions to \SETH\ that can be used to prove our results.
We present our framework for obtaining quadratic lower bounds in \secref{framework}. 
We then first prove a conditional lower bound for LCS in \secref{lcs}; this proof is superseded by the conditional lower bound for edit distance in \secref{edit}, but it is shorter and might be more accessible. 
Quadratic-time hardness of dynamic time warping follows in \secref{dtw}.
Finally, in \secref{LPSLTS} we prove hardness of longest palindromic subsequence and longest tandem subsequence.

\section{Preliminaries} \label{sec:preliminaries}

For a sequence $x$, we write $|x|$ for its length, $x[k]$ for its $k$-th entry, $x[k..\ell]$ for the substring from $x[k]$ to $x[\ell]$, and $\rev(x)$ for the reversed sequence. For sequences $x,y$ we denote their concatenation by $x \, y$. A \emph{traversal} of two sequences $x,y$ of length $n,m$, respectively, is a sequence of pairs $( (a_1,b_1),\ldots,(a_t,b_t) )$ with $t \in \mathbb{N}$ satisfying (1) $(a_1,b_1) = (1,1)$, (2) $(a_t,b_t) = (n,m)$, and (3) $(a_{i+1},b_{i+1})$ is either of $(a_i+1,b_i)$, $(a_i,b_i+1)$, or $(a_i+1,b_i+1)$ for all $1 \le i < t$. 

\paragraph{Edit Distance} Let $x,y$ be strings over an alphabet $\Sigma$ of length $n,m$ ($n \ge m$), respectively. For a traversal $T = ( (a_1,b_1),\ldots,(a_t,b_t) )$ of $x,y$ we say that its $i$-th operation, $1 \le i < t$, is (1) a \emph{deletion in $x$} if $(a_{i+1},b_{i+1}) = (a_i+1,b_i)$, (2) a \emph{deletion in $y$} if $(a_{i+1},b_{i+1}) = (a_i,b_i+1)$, (3) a \emph{matching} if $(a_{i+1},b_{i+1}) = (a_i+1,b_i+1)$ and $x[a_i] = y[b_i]$, or (4) a \emph{substitution} if $(a_{i+1},b_{i+1}) = (a_i+1,b_i+1)$ and $x[a_i] \ne y[b_i]$. These four operations incur costs of $\cdelx, \cdely, \cmatch$, and $\csubst$, respectively. We will always assume that these costs are rational constants, so that we can ignore representation issues.
The cost $\dEDIT(T)$ of a traversal $T$ is the total cost of all its operations. The edit distance $\dEDIT(x,y)$ is the minimal cost of any traversal of $x,y$. We write $\EDITallc$ for the problem of computing the edit distance of two given strings with costs  $\cdelx, \cdely, \cmatch$, and $\csubst$. We write $\EDITc$ as a shorthand for $\EDIT(1,1,0,\csubst)$. Note that for these problems the costs of all four operations are constant, i.e., they stay fixed with growing $n,m$. We will mostly consider edit distance over binary strings, i.e., we set $\Sigma = \{0,1\}$.

\paragraph{Dynamic Time Warping (\DTW)}
Let $(M,d)$ be any metric space. Let $x,y$ be curves, i.e., sequences over $M$ of length $n,m$ ($n \ge m$), respectively. The cost $\dDTW(T)$ of a traversal $T = ( (a_1,b_1),\ldots,(a_t,b_t) )$ is $\sum_{i=1}^t d(x[a_i],y[b_i])$. The dynamic time warping distance $\dDTW(x,y)$ is the minimal cost of any traversal of $x$ and $y$. 
We obtain the special case of \emph{dynamic time warping on one-dimensional curves} by setting $M = \mathbb{R}$ and $d(a,b) := |a-b|$ for any $a,b \in \mathbb{R}$.

\subsection{Hardness Assumptions}

Consider the \emph{Orthogonal Vectors problem} (\OV): Given sets $A,B$ of vectors in $\{0,1\}^d$, $|A|=n, |B|=m$, decide whether there is a pair of vectors $a \in A, b \in B$ such that $a[k]\cdot b[k] = 0$ for all $k$ (which we denote by $\langle a,b\rangle=0$). Clearly, this problem can be solved in time $\Oh(n^2 d)$. The best-known algorithm runs in time $n^{2-1/\Oh(\log(d/\log n))}$~\cite{abboud2015polynomial}, which is only slightly subquadratic for $d \gg \log n$. Thus, the following hypotheses are reasonable.

\paragraph{\textup{\textbf{Orthogonal Vectors Hypothesis}} (\OVH):} \emph{For no $\eps > 0$ there is an algorithm for OV, restricted to $n=m$, that runs in time $\Oh(n^{2-\eps} \poly(d))$.} 

\paragraph{\textup{\textbf{Unbalanced Orthogonal Vectors Hypothesis}} (\UOVH):} \emph{Let $0 < \alpha \le 1$. For no $\eps > 0$ there is an algorithm for OV, restricted to $m= \Theta(n^\alpha)$ and $d \le n^{o(1)}$, that runs in time $\Oh((nm)^{1-\eps})$.} 

\paragraph{}
It is well-known that \SETH\ implies \OVH~\cite{williams05}. A slight generalization shows that \SETH\ also implies \UOVH. Hence, these hypotheses are weaker assumptions than \SETH.

\begin{lem}
  \SETH\ implies \OVH\ and \UOVH.
\end{lem}
\begin{proof}
  For \OVH\ the statement follows from~\cite{williams05}.
  Let $0 < \eps < 1/2$ and $0 < \alpha \le 1$. Assume that Orthogonal Vectors, restricted to $m = \Theta(n^\alpha)$ and $d \le n^{o(1)}$, has an $\Oh((nm)^{1-\eps})$ algorithm. We show that this contradicts \SETH.
  To this end, let $\varphi$ be an instance of \kSAT\ with $N$ variables and $M$ clauses. We use the sparsification lemma~\cite{ImpagliazzoPZ01}, which yields $t := 2^{\eps N/2}$ \kSAT\ instances $\varphi_1,\ldots,\varphi_{t}$ with $N$ variables and $f(k,\eps)\cdot N$ clauses such that $\varphi$ is satisfiable if and only if some $\varphi_i$ is satisfiable. 
  If $N \le f(k,\eps)$ then we decide each $\varphi_i$ in time $\Oh_{k,\eps}(1)$. Otherwise, $\varphi_i$ has at most $N^2$ clauses, and we can assume equality by duplicating clauses. In this case, we construct an instance of Orthogonal Vectors as follows. 
  Let $x_1,\ldots,x_N$ be the variables and $C_1,\ldots,C_{N^2}$ be the clauses of $\varphi_i$. We set $d := N^2$ and split the variables into the left half $x_1,\ldots,x_{N/(1+\alpha)}$ and the right half $x_{N/(1+\alpha)+1},\ldots,x_N$. The set $A$ consists of one vector $a_z \in \mathbb{R}$ for every assignment $z$ of true and false to the left half of the variables. If $z$ causes clause $C_i$ to be true, i.e., some unnegated variable of $C_i$ is set to true in $z$ or some negated variable of $C_i$ is set to false in $z$, then we set $a_z[i] := 0$. Otherwise, we set $a_z[i] := 1$. Similarly, set $B$ has a vector $b_{z'}$ for any assignment $z'$ of true or false to the right half of the variables and $b_{z'}[i] = 0$ or 1, depending on whether $z'$ causes clause $C_i$ to be true. Then $\langle a_{z},b_{z'}\rangle=0$ if and only if $(z,z')$ forms a satisfying assignment of $\varphi_i$. Thus, we can decide $\varphi_i$ by solving the constructed instance of Orthogonal Vectors. Note that $n = |A| = 2^{N/(1+\alpha)}$ and $m = |B| = 2^{N \alpha/(1+\alpha)}$, so that indeed $m = \Theta(n^\alpha)$. Moreover, $d = N^2 \le 2^{o(N)} = n^{o(1)}$. Thus, we can apply the algorithm for Orthogonal Vectors, that we assumed to exist, running in time $\Oh((nm)^{1-\eps}) = \Oh(2^{(1-\eps)N})$. Running this procedure for all $\varphi_i$ decides $\varphi$ in time $\Oh(t \cdot 2^{(1-\eps)N}) = \Oh(2^{(1-\eps/2)N})$, contradicting \SETH.
\end{proof}

Thus, any lower bound conditional on \OVH\ or \UOVH\ also holds conditional on \SETH. In fact, we prove all of our results by reductions from Orthogonal Vectors, so that in our results we may replace the assumption \SETH\ by \OVH\ or \UOVH. Specifically, in \thmrefss{dtw}{EDIT}{LPSLTS} we can replace \SETH\ by \OVH, and in \thmref{unbalanced} we can replace \SETH\ by \UOVH.
We remark that a version of \OVH\ has also been used in~\cite{abboud_quadratic_2015} and is implicit in many other \SETH-based lower bounds.

\section{Framework} \label{sec:framework}

We consider a similarity (or distance) measure $\delta:\inputs\times\inputs\to\mathbb{N}_0$, where $\inputs$ denotes the set of inputs, e.g.,
all binary strings or all one-dimensional curves. By a reduction from Orthogonal Vectors, we prove that computing this similarity measure cannot
be done in strongly subquadratic time unless SETH fails if $\delta$ admits a gadget
that allows us to exactly realize alignments of inputs $x_{1},\dots,x_{n}\in\inputs$ and $y_{1},\dots,y_{m}\in\inputs$.
To formally state the requirement, we start by introducing the following notions.

\paragraph{Types} In this paper, we define the \emph{type} of a sequence $x \in \inputs$ to be its length and the sum of its entries, i.e., $\type(x) := (|x|, \sum_i x[i])$ (where for binary strings $\sum_k x[k]$ is to be interpreted as the number of ones in $x$). The definition of types can be customized to the similarity measure under consideration and is chosen to work for the problems considered in this paper.
We define $\inputs_{t}:=\{x \in \inputs\mid\type(x)=t\}$ as the
set of inputs of type $t$. 

\paragraph{Alignments}
Let $n \ge m$. A \emph{(partial) alignment} is a set $A = \{(i_1,j_1),\ldots,(i_k,j_k)\}$ with $0 \le k \le m$ such that $1 \le i_1 < \ldots < i_k \le n$ and $1 \le j_1 < \ldots < j_k \le m$. We say that $(i,j) \in A$ are \emph{aligned}. Any $i \in [n]$ or $j \in [m]$ that is not contained in any pair in $A$ is called \emph{unaligned}.
We denote the set of all partial alignments (with respect to $n,m$) by $\algn_{n,m}$.

We call the partial alignment $\{(\Delta+1,1),\ldots,(\Delta+m,m)\}$, with $0 \le \Delta \le n-m$, a \emph{structured alignment}. We denote the set of all structured alignments by $\strc_{n,m}$.

For any $x_1,\ldots,x_n \in \inputs$ and $y_1,\ldots,y_m \in \inputs$ we define the \emph{cost} of alignment $A \in \algn_{n,m}$ by 
$$ \delta(A) = \delta^{x_1,\ldots,x_n}_{y_1,\ldots,y_m}(A) := \sum_{(i,j) \in A} \delta(x_i,y_j) + (m-|A|) \max_{i,j} \delta(x_i,y_j). $$
In other words, for any $j \in [m]$ which is aligned to some $i$ we pay the distance $\delta(x_i,y_j)$, while for any unaligned $j$ we pay the maximal distance of any $(x_{i'},y_{j'})$ (note that there are $m-|A|$ unaligned $j \in [m]$, see \figref{alignments}). This means that we get punished for any unaligned~$j$. 

\begin{figure}

\begin{subfigure}{0.49\textwidth}
  \includegraphics[width=\textwidth]{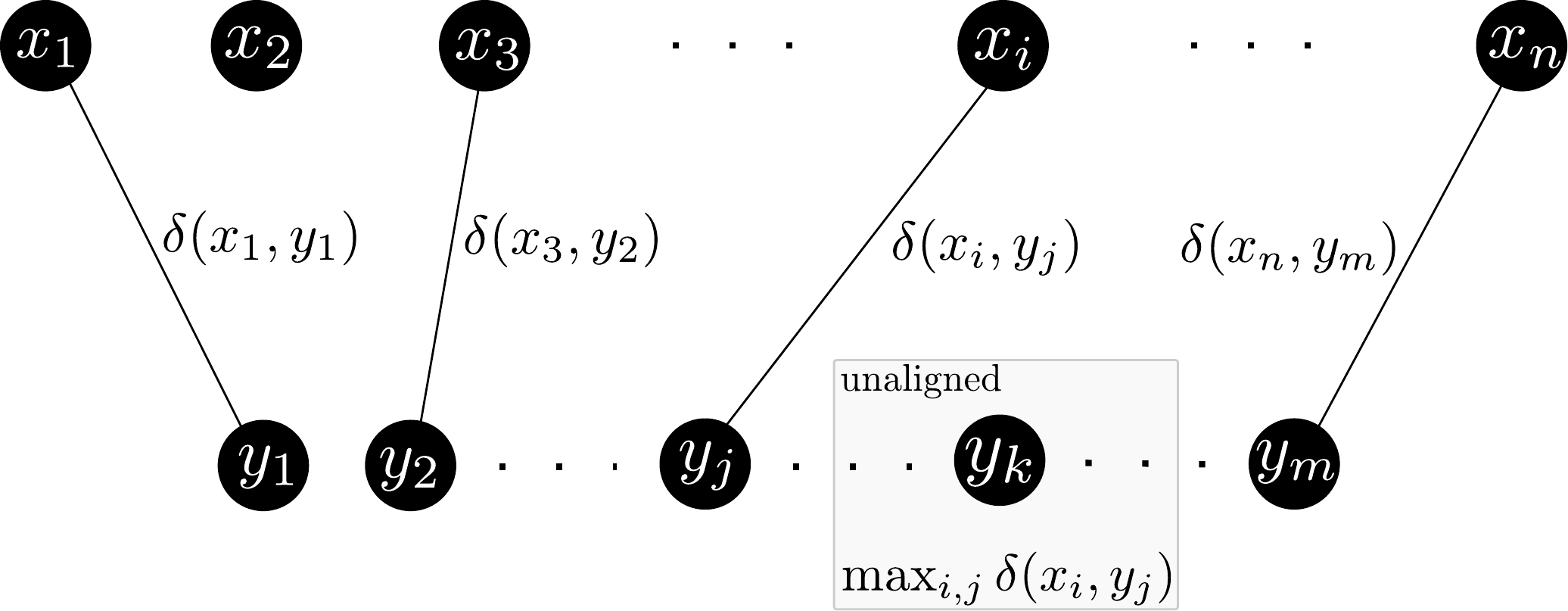}
  \caption{Cost $\delta(A)$ of a partial alignment $A\in \algn_{n,m}$}
  \label{fig:partialA}
\end{subfigure}
\qquad
\begin{subfigure}{0.49\textwidth}
  \includegraphics[width=\textwidth]{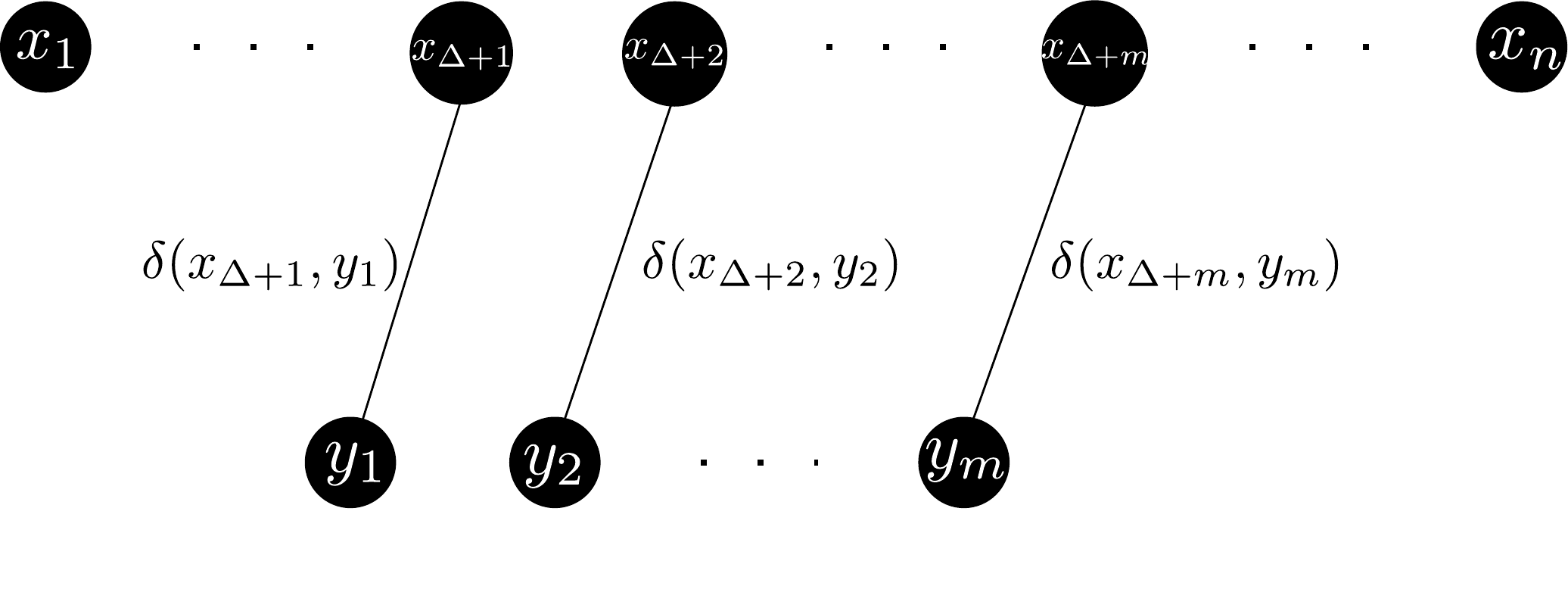}
  \caption{Cost $\delta(A)$ of a structured alignment $A\in \strc_{n,m}$}
  \label{fig:structuredA}
\end{subfigure}

\caption{Costs of Alignments}
\label{fig:alignments}
\end{figure}

\paragraph{\Aligngad\ Gadget}
We start with some intuition. Consider the problem of computing the value $\min_{A \in \strc_{n,m}} \Val{A}$. This can be solved in time $\Oh(nm)$ if each $\delta(x_i,y_j)$ can be evaluated in constant time, since $|\strc_{n,m}| = \Oh(n)$ and evaluating $\Val{A}$ amounts to computing $m$ values $\delta(x_i,y_j)$. Moreover, intuitively it should not be possible to compute this value in strongly subquadratic time.
We will show that in some sense it is even hard to compute, in strongly subquadratic time, \emph{any value} $v$ with
\begin{align}
  \min_{A \in \algn_{n,m}} \Val{A} \le v \le \min_{A \in \strc_{n,m}} \Val{A}.  \label{Aineqs}
\end{align}
Now, an \aligngad\ gadget is simply a pair of instances $(x,y)$ such that from $\delta(x,y)$ we can infer\footnote{For us ``infer'' will simply mean that $v = \delta(x,y)-C$ for an appropriate $C$.} a value $v$ as above. The main reason to relax our goal from computing $\min_{A \in \strc_{n,m}} \Val{A}$ to satisfying \eqref{Aineqs} is that this makes constructing \aligngad\ gadgets much easier.
Note that for the \aligngad\ gadget $(x,y)$ computing $\delta(x,y)$ is as hard as computing $\min_{A \in \strc_{n,m}} \Val{A}$ (in an approximate sense as given by \eqref{Aineqs}), which we argued above should take quadratic time. This informal discussion motivates the following definition.

\begin{defn}
\label{def:cagadget}
The similarity measure $\delta$ admits an 
\emph{\aligngad\ gadget, }if the following
conditions hold: 
Given instances $x_{1},\dots,x_{n}\in\inputs_{t_{\x}}$, $y_{1},\dots,y_{m}\in\inputs_{t_{\y}}$ with $m\le n$ and types $t_\x=(\ell_\x,s_\x),t_\y=(\ell_\y,s_\y)$, we can
construct new instances
$x=\CandA_{\x}^{m,t_{\y}}(x_{1},\dots,x_{n})$ and $y=\CandA_{\y}^{n,t_{\x}}(y_{1},\dots,y_{m})$
and $C\in\mathbb{Z}$ such that 
\begin{align} \label{eq:Cone}
\min_{A \in \algn_{n,m}} \Val{A} \le \delta(x,y) - C \le \min_{A \in \strc_{n,m}} \Val{A}.
\end{align}
Moreover, $\type(x)$ and $\type(y)$ only depend on $n,m,t_\x$, and $t_\y$. Finally, this construction runs in time $\bigOh((n+m)(\ell_\x + \ell_\y))$.

If the construction additionally fulfills $|x|=\Oh(n(\ell_\x + \ell_\y))$
and $|y|=\Oh(m(\ell_\x + \ell_\y))$, then we say that $\delta$
admits an \emph{unbalanced \aligngad\ gadget}.
\end{defn}

Note that the types serve the purpose of simplifying the algorithmic
problem in the above definition by restricting the inputs
to same-type objects. If we can construct suitable $x$ and $y$ for
\emph{arbitrary} inputs $x_{1},\dots,x_{n}$ and $y_{1},\dots,y_{m}$ then
we may completely disregard types. 

\begin{defn}\label{def:coordValues}
The similarity measure $\delta$ admits \emph{coordinate values}, if there exist $\zleft,\zright,\oleft,\oright\in\inputs$ satisfying
\[
\delta(\oleft,\oright)>\delta(\zleft,\oright)=\delta(\zleft,\zright)=\delta(\oleft,\zright),
\]
and moreover, $\type(\zleft)=\type(\oleft)$ and $\type(\zright)=\type(\oright)$.
\end{defn}

\begin{thm} \label{thm:main}
Let $\delta$ be a similarity measure admitting an \aligngad\
gadget and coordinate values and consider the problem of computing $\delta(x,y)$ with $|x| \le n$, $|y| \le m$, and $m \le n$. For no $\eps > 0$ this problem can be solved in time $\bigOh(m^{2-\varepsilon})$ unless \OVH\ fails.
If $\delta$ even admits an unbalanced \aligngad\ gadget, then
for no $\eps > 0$ this problem can be solved in time $\bigOh((n m)^{1-\varepsilon})$, unless \UOVH\ fails. Both statements hold restricted to $n^{\alpha - o(1)} \le m \le n^{\alpha + o(1)}$ for any $0<\alpha \le 1$.
\end{thm}

\subsection{Proof of \thmref{main}}

We present a reduction from \OV\ to the problem of computing $\delta$. This uses constructions and arguments similar to \cite{bringmann_walking_2014,Arthurs}.
Consider an instance $a_{1},\dots,a_{n}\in\{0,1\}^{d}$
and $b_{1},\dots,b_{m}\in\{0,1\}^{d}$ of \OV, $n \ge m$. We construct $x,y\in\inputs$
and $\rho\in\mathbb{N}_0$ such that $\delta(x,y)\le\rho$ if and only
if there are $i \in [n]$ and $j \in [m]$ with $\langle a_{i},b_{j}\rangle=0$. 
To this end, let $a_{i}[k]$ denote the $k$-th component of $a_{i}$. For all
$i \in [n]$ and $j \in [m]$, we construct \emph{coordinate gadgets} as follows
\begin{align*}
\CG(a_{i},k) & :=\begin{cases}
\zleft & \text{if }a_{i}[k]=0\\
\oleft & \text{if }a_{i}[k]=1
\end{cases}\quad1\le k\le d, & \CG(a_{i},d+1) & :=\zleft,\\
\CG(b_{j},k) & :=\begin{cases}
\zright & \text{if }b_{j}[k]=0\\
\oright & \text{if }b_{j}[k]=1
\end{cases}\quad1\le k\le d, & \CG(b_{j},d+1) & :=\oright.
\end{align*}
Note that we have $\type(\CG(a_{i},1))=\cdots=\type(\CG(a_{i},d+1)) =: t_\x$ and $\type(\CG(b_{j},1))=\cdots=\type(\CG(b_{j},d+1)) =: t_\y$ by definition of coordinate values.
This allows us to use the alignment gadget to obtain the following \emph{vector gadgets}
\begin{eqnarray*}
\VG(a_{i}) & := & \CandA_{\x}^{d+1,t_{\y}}(\CG(a_{i},1),\dots,\CG(a_{i},d+1)),\\
\VG(b_{j}) & := & \CandA_{\y}^{d+1,t_{\x}}(\CG(b_{j},1),\dots,\CG(b_{j},d+1)),\\
S & := & \CandA_{\x}^{d+1,t_{\y}}(\underbrace{\zleft,\dots,\zleft,\oleft}_{d+1}),
\end{eqnarray*}
Note that $\type(\VG(a_{1}))= \ldots = \type(\VG(a_{n}))=\type(S) =: t_\x'$
and $\type(\VG(b_{1})) = \ldots = \type(\VG(b_{m})) =: t_\y'$, because the type of the output of the alignment gadget only depends on the number of input elements and their type, which are all $t_\x$ or all $t_\y$, respectively. We introduce \emph{normalized vector gadgets} as follows
\begin{eqnarray*}
\NVG(a_{i}) & := & \CandA_{\x}^{1,t_\y'}(S,\VG(a_{i})),\\
\NVG(b_{j}) & := & \CandA_{\y}^{2,t_\x'}(\VG(b_{j})).
\end{eqnarray*}
Note that we have $\type(\NVG(a_{1})) = \ldots = \type(\NVG(a_{n})) =: t_\x''$ and $\type(\NVG(b_{1})) = \ldots = \type(\NVG(b_{m})) =: t_\y''$.
We finally obtain $x$ and $y$ by setting
\begin{eqnarray*}
x & := & \CandA^{m,t_\y''}_{\x}(\NVG(a_{1}),\dots,\NVG(a_{n}),\NVG(a_{1}),\dots,\NVG(a_{n})),\\
y & := & \CandA^{2n,t_\x''}_{\y}(\NVG(b_{1}),\dots,\NVG(b_{m})).
\end{eqnarray*}
We denote by $C, C', C''$ the value $C$ in the three invocations of Property~\eqref{eq:Cone} of the alignment gadget.

Observe that $x$ and $y$ have length $\Oh((n+m)d)$ and can be constructed
in time $\bigOh((n+m)d)$ by applying the algorithm implicit in \defref{cagadget}
three times. Moreover, if $\delta$ admits an \emph{unbalanced} \aligngad\ gadget, then we have $|x| = \Oh(nd)$ and $|y| = \Oh(md)$. It remains to show that if we know $\delta(x,y)$ then we can decide the given \OV\ instance in constant time, i.e., correctness of our construction, which we do below. This finishes our reduction from \OV\ to the problem of computing $\delta$. 
To obtain \thmref{main}, let $0 < \alpha \le 1$ and assume that $\delta(x',y')$ can be computed in time $\Oh(M^{2-\eps})$ whenever $|x'| \le N$, $|y'| \le M$, and $N^{\alpha-o(1)} \le M \le N^{\alpha+o(1)}$. Then in particular for $n=m$ we can compute $\delta(x,y)$ in time $\Oh(\min\{|x|,|y|\}^{2-\eps} + |x| + |y|) = \Oh(((n+m)d)^{2-\eps}) = \Oh((nd)^{2-\eps})$, contradicting \OVH. In case of an unbalanced \aligngad\ gadget, assume that $\delta(x',y')$ can be computed in time $\Oh((N M)^{1-\eps})$ whenever $|x'| \le N$, $|y'| \le M$, and $N^{\alpha-o(1)} \le M \le N^{\alpha+o(1)}$. Then for $m = \Theta(n^\alpha)$ and $d \le n^{o(1)}$ we can compute $\delta(x,y)$ in time $\Oh((|x| |y|)^{1-\eps} + |x|+|y|) = \Oh( ((nd) (md))^{1-\eps} + (n+m)d) = \Oh((nm)^{1-\eps/2})$, contradicting \UOVH. This proves \thmref{main}.

\paragraph{Correctness}
We now prove correctness of our construction and refer to \figref{NVGs} for an intuition for coordinate, vector, and normalized vector gadgets. Let $\rho_{0}:=\delta(\zleft,\zright)=\delta(\zleft,\oright)=\delta(\oleft,\zright)$ and $\rho_{1}:=\delta(\oleft,\oright)$. Recall that $\rho_0 < \rho_1$.

\begin{figure}

\begin{subfigure}{0.47\textwidth}
  \includegraphics[width=\textwidth]{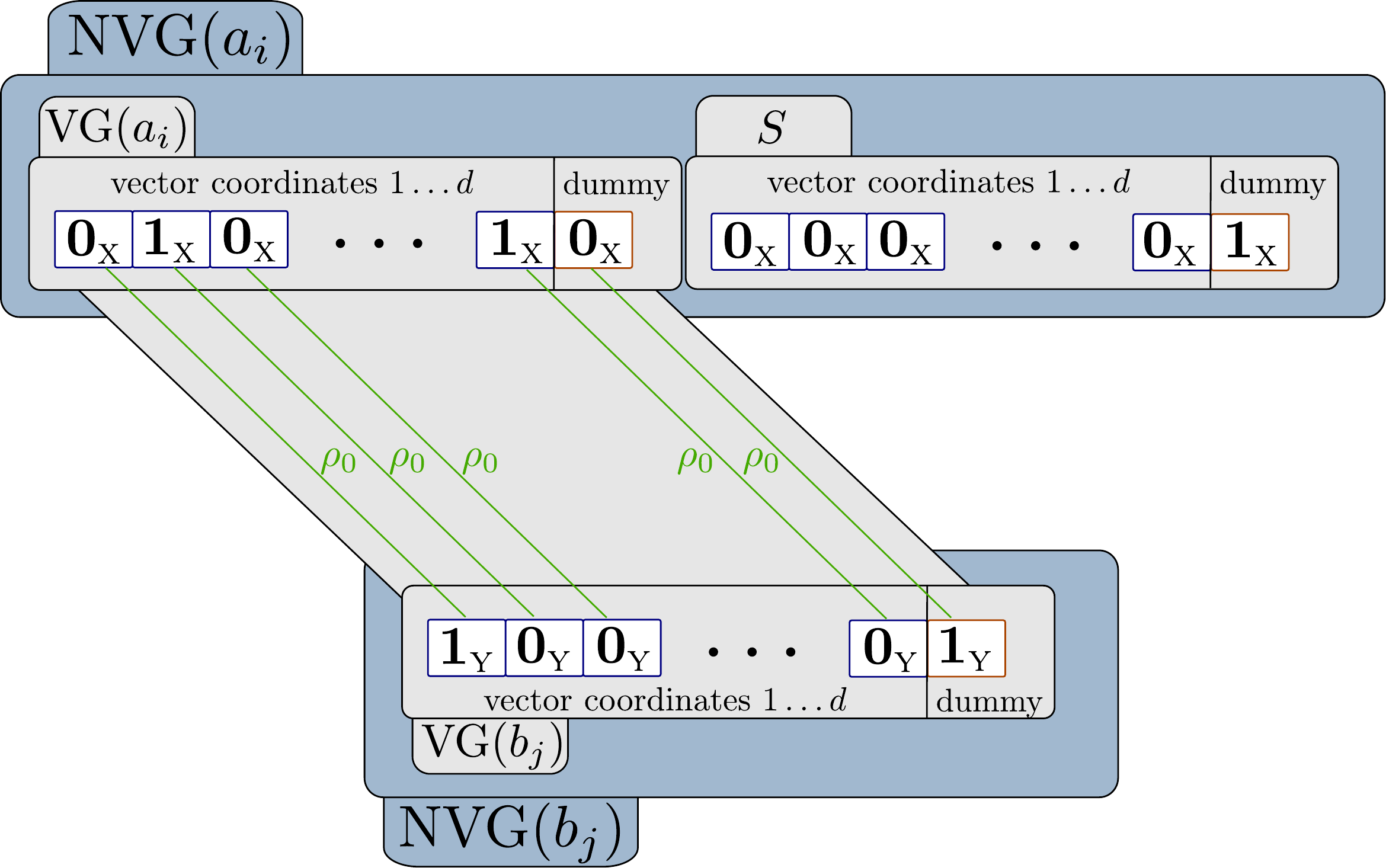}
  \caption{Case $\langle a_i, b_j \rangle = 0$. Aligning $\VG(b_j)$ with $\VG(a_i)$ achieves an alignment cost of $(d+1)\rho_0$.}
  \label{fig:NVGs-sat}
\end{subfigure}
\qquad
\begin{subfigure}{0.47\textwidth}
  \includegraphics[width=\textwidth]{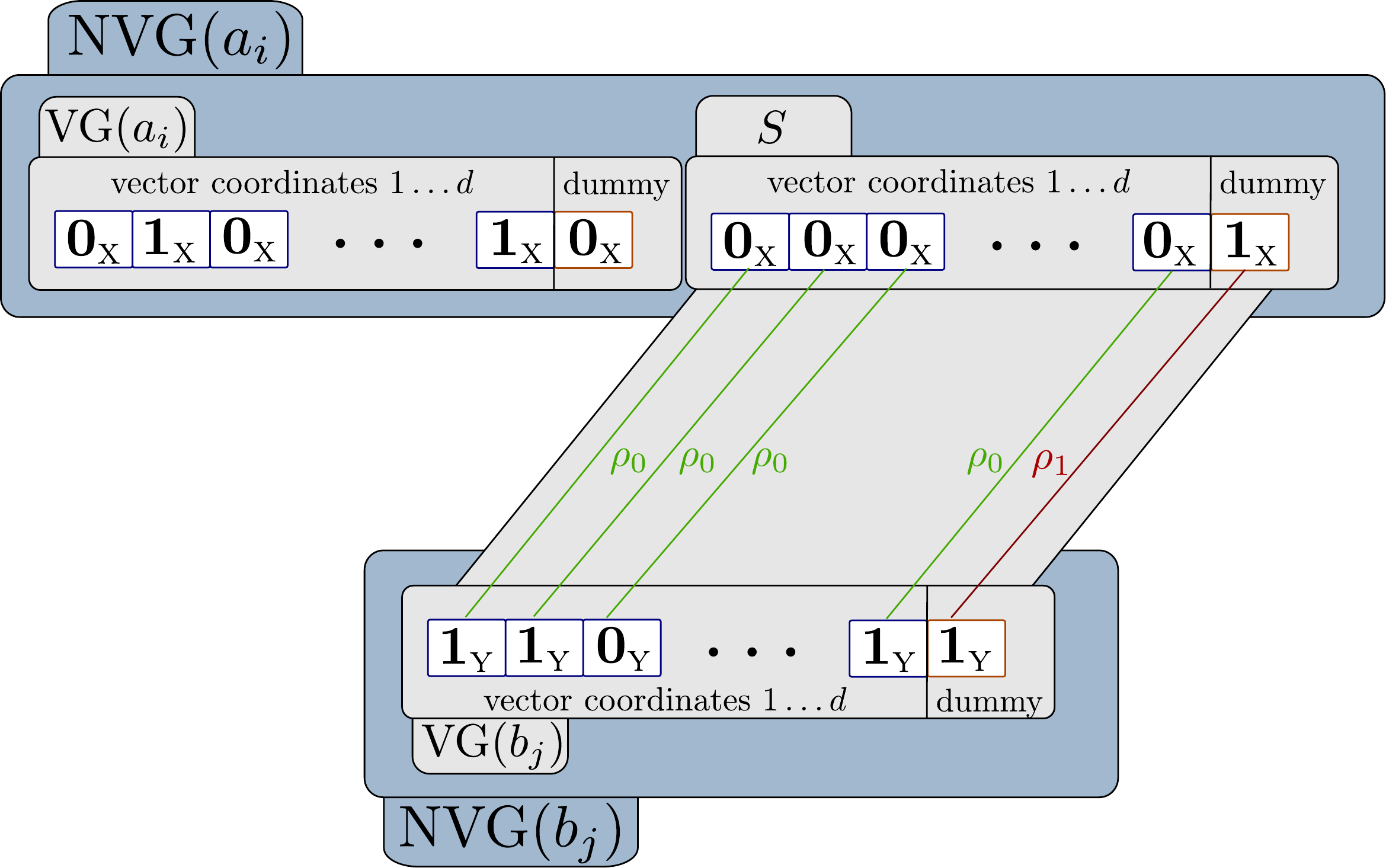}
  \caption{Case $\langle a_i, b_j \rangle > 0$. Aligning $\VG(b_j)$ with $S$ achieves an alignment cost of $d\rho_0 + \rho_1$.}
  \label{fig:NVGs-unsat}
\end{subfigure}

\caption{Schematic illustration of the coordinate, vector, and normalized vector gadgets.}
\label{fig:NVGs}
\end{figure}

\begin{claim} \label{cla:AG}
For any $i \in [n]$, $j \in [m]$, if $\langle a_{i},b_{j}\rangle=0$, then $\delta(\VG(a_{i}),\VG(b_{j})) = C+(d+1)\rho_{0}$.
Otherwise, $\delta(\VG(a_{i}),\VG(b_{j}))\ge C+d\rho_{0}+\rho_{1}$. Moreover, $\delta(S,\VG(b_j)) = C+d\rho_0+\rho_1$.
\end{claim}
\begin{proof}
  If $\langle a_{i},b_{j}\rangle=0$, then the structured alignment $\{(1,1),\ldots,(d+1,d+1)\}$ has cost $\delta(A) = \sum_{k=1}^{d+1} \delta(\CG(a_i,k),\CG(b_j,k)) = (d+1) \rho_0$, since in each component $k$ at least one value is $\zleft$ or $\zright$, incurring a cost of $\rho_0$ (indeed even in position $k=d+1$ we have $\CG(a_i,d+1)=\zleft$). By definition of alignment gadgets, we obtain $\delta(\VG(a_{i}),\VG(b_{j})) - C \le (d+1)\rho_0$. Moreover, since the cost $\delta(A)$ of any alignment $A \in \algn_{d+1,d+1}$ consists of $d+1$ summands of the form $\delta(u_\x,u_\y)$ with $u_\x \in \{\zleft,\oleft\}, u_\y \in \{\zright,\oright\}$, we also have $\delta(\VG(a_i),\VG(b_j))-C \ge (d+1) \rho_0$.
  
  If $\langle a_{i},b_{j}\rangle>0$, then consider any $A \in \algn_{n,m}$. If $|A|=d+1$ then $A = \{(1,1),\ldots,(d+1,d+1)\}$, and this alignment incurs a cost of at least $d \rho_0 + \rho_1$, since in at least one position $k$ we have $\CG(a_i,k) = \oleft$ and $\CG(b_j,k) = \oright$. Otherwise, if $|A| < d+1$, then $\delta(A)$ consists of $d+1$ summands of the form $\delta(u_\x,u_\y)$ with $u_\x \in \{\zleft,\oleft\}, u_\y \in \{\zright,\oright\}$, and at least one of these summands is the punishment term $\max_{k,\ell} \delta(\CG(a_i,k),\CG(b_j,\ell))$ because $|A| < d+1$. Since $\langle a_{i},b_{j}\rangle=1$, the punishment term is $\rho_1$ and we obtain $\delta(A) \ge d \rho_0 + \rho_1$. By definition of alignment gadgets, we have $\delta(\VG(a_i),\VG(b_j))-C \ge d \rho_0 + \rho_1$.
  
  We argue similarly for $\delta(S,\VG(b_j))$: The alignment $\{(1,1),\ldots,(d+1,d+1)\}$ incurs a cost of $d\rho_0+\rho_1$, since the $(d+1)$-th component of $S$ is $\oleft$ and of $\VG(b_j)$ is $\VG(b_j,d+1) = \oright$, while all other components of $S$ are $\zleft$. Moreover, any alignment with $|A| < d+1$ incurs a punishment term, so that it incurs cost of at least $d\rho_0+\rho_1$.
\end{proof}

\begin{claim}
For any $i \in [n]$, $j \in [m]$, if $\langle a_{i},b_{j}\rangle=0$ then $\delta(\NVG(a_{i}),\NVG(b_{j})) = C+C'+(d+1)\rho_{0}=:\rho'_{0}$.
Otherwise, $\delta(\NVG(a_{i}),\NVG(b_{j}))=C+C'+d\rho_{0}+\rho_{1}=:\rho'_{1}$.
\end{claim}
\begin{proof}
Note that $\{(1,1)\}$, $\{(2,1)\}$, and $\emptyset$ are the only alignments in $\algn_{2,1}$, which corresponds to aligning
$(S,\VG(b_{j}))$ or $(\VG(a_{i}),\VG(b_{j}))$ or nothing. Moreover, the structured alignments are $\{(1,1)\}$ and $\{(2,1)\}$. Observe that the cost of the alignment $\emptyset$ is simply the maximum of the other two alignments. By \claref{AG}, if $\langle a_{i},b_{j}\rangle=0$ then the minimal cost is $C+(d+1)\rho_0$, attained by alignment $\{(2,1)\}$. Otherwise, the minimal cost is $C+d \rho_0 + \rho_1$, attained by alignment $\{(1,1)\}$. By definition of alignment gadgets, this yields
that $\delta(\NVG(a_i),\NVG(b_j)) - C'$ is equal to $C+(d+1)\rho_0$ or $C+d \rho_0 + \rho_1$, respectively.
\end{proof}

The claim shows that $\delta(\NVG(a_{i}),\NVG(b_{j}))$ attains one
of only two values, depending on whether $\langle a_{i},b_{j}\rangle=0$.

\begin{claim}
If there is no $i \in [n], j \in [m]$ with $\langle a_{i},b_{j}\rangle=0$,
then $\delta(x,y)\ge C''+m\rho'_{1}$. Otherwise, $\delta(x,y)\le C''+(m-1)\rho'_{1}+\rho'_{0}$.\end{claim}
\begin{proof}
If $\langle a_{i},b_{j}\rangle>0$ for all $i,j$, then by the previous claim we have $\delta(\NVG(a_i),\NVG(b_j)) \ge \rho_1'$ for all $i,j$. Since the cost of any alignment consists of $m$ summands of the form $\delta(\NVG(a_i),\NVG(b_j))$ for some $i,j$, the cost of any alignment is at least $m \rho_1'$. By definition of alignment gadgets, we obtain $\delta(x,y) - C'' \ge m  \rho_1'$.

If $\langle a_{i},b_{j}\rangle=0$ for some $i,j$, then consider the structured alignment $A = \{(\Delta+1,1),\ldots,(\Delta+m,m)\}$ with $\Delta := i-j$ if $i\ge j$, or $\Delta := n+i-j$ if $i < j$. Its cost consists of $m$ summands, of which one is $\delta(\NVG(a_i),\NVG(b_j)) = \rho_0'$ and all others are at most $\rho_1'$. Hence, the cost of $A$ is at most $(m-1)\rho_1' + \rho_0'$ and by definition of alignment gadgets, we obtain $\delta(x,y) - C'' \le (m-1)\rho_1' + \rho_0'$.
\end{proof}
By setting $\rho := C''+(m-1)\rho'_{1}+\rho'_{0}$ we have found a threshold such that $\delta(x,y)\le\rho$ if and only
if there is a pair $(i,j)$ with $\langle a_{i},b_{j}\rangle=0$. Thus, computing $\delta(x,y)$ allows to decide the given \OV\ instance. This finishes the proof of \thmref{main}.

\section{Longest Common Subsequence} \label{sec:lcs}

In this section, we present an alternative hardness proof for longest common subsequence (\LCS), which is shorter than for the more general problem \EDITallc\ in \secref{edit}. 
Given two strings $x,y$ over an alphabet $\Sigma$, a longest common subsequence is a binary string $z$ that appears in $x$ and in $y$ as a subsequence and has maximal length. We denote by $\LCS(x,y)$ some longest common subsequence of $x$ and $y$, and by $|\LCS(x,y)|$ the length of any longest common subsequence of $x$ and $y$. 

We present an \aligngad\ gadget and coordinate values for \LCS\ over binary strings, i.e., we consider the set of inputs ${\cal I} := \bigcup_{k \ge 0} \{0,1\}^k$.
Note that \LCS\ is a maximization problem, but \defref{cagadget} implicitly assumes a minimization problem, so we instead consider the number of unmatched symbols $\dLCS(x,y) := |x|+|y|-2|\LCS(x,y)|$ for binary strings $x,y$. Note that this is equivalent to \EDITallc\ for $\cdelx = \cdely = 1$, $\cmatch = 0$, and $\csubst = 2$.

\begin{lem}
\LCS\ admits coordinate values by setting
$$ \oleft := 11100,\; \zleft := 10011,\; \oright := 00111,\; \zright := 11001. $$
\end{lem}
\begin{proof}
All four values have the same length and the same number of 1s, so they have equal type. Short calculations show that $\LCS(\oleft,\oright) = 111$,  $\LCS(\oleft,\zright) = 1100$, $\LCS(\zleft,\oright) = 0011$, and $\LCS(\zleft,\zright) = 1001$. Thus, 
  $4 = \dLCS(\oleft,\oright) > \dLCS(\oleft,\zright) = \dLCS(\zleft,\oright) = \dLCS(\zleft,\oright) = 2$.
\end{proof}

\begin{defn} \label{def:lcs}
Consider instances $x_1,\ldots,x_n \in \inputs_{t_\x}$ and $y_1,\ldots,y_m \in \inputs_{t_\y}$ with $n \ge m$ and types $t_\x = (\ell_\x,s_\x), t_\y = (\ell_\y,s_\y)$. Set $\gamma_1 := \ell_\x+\ell_\y, \gamma_2 := 6(\ell_\x+\ell_\y), \gamma_3 := 10(\ell_\x+\ell_\y)+2s_\x-\ell_\x, \gamma_4 := 13(\ell_\x+\ell_\y)$.
We \emph{guard} the input strings by blocks of zeroes and ones, setting $\guard(z) := 1^{\gamma_2} 0^{\gamma_1} z 0^{\gamma_1} 1^{\gamma_2}$. 
We define the \aligngad\ gagdet as 
\begin{align*}
x &:=\qquad \, \guard(x_1) \; 0^{\gamma_3} \; \guard(x_2) \; 0^{\gamma_3} \,\ldots\, \guard(x_{n-1}) \; 0^{\gamma_3} \; \guard(x_n),  \\
y &:= \, 0^{n\gamma_4} \; \guard(y_1) \; 0^{\gamma_3} \; \guard(y_2) \; 0^{\gamma_3} \,\ldots\, \guard(y_{m-1}) \; 0^{\gamma_3} \; \guard(y_m) \; 0^{n\gamma_4}.
\end{align*}
\end{defn}

\begin{lem} \label{lem:lcscorrect}
  \defref{lcs} realizes an \aligngad\ gadget for \LCS.
\end{lem}

Thus, \thmref{main} is applicable, implying a lower bound of $\Oh(m^{2-\eps})$ for \LCS. 
We remark that our construction is no \emph{unbalanced} \aligngad\ gadget, as the length of $y$ grows linearly in $n$, not necessarily in $m \le n$. Thus, we do not obtain a conditional lower bound of $\Oh((n m)^{1-\eps})$ (for $m \approx n^\alpha$ for any $0 < \alpha < 1$).

\begin{proof}[Proof of \lemref{lcscorrect}]
  Observe that indeed $x$ only depends on $m,t_\y$, and $x_1,\ldots,x_n$, and $\type(x)$ only depends on $n,m,t_\x$, and $t_\y$, and similarly for $y$. Moreover, $x$ and $y$ can clearly be constructed in time $\Oh((n+m)(\ell_\x+\ell_\y))$, where $\ell_\x = |x_1| = \ldots = |x_n|$ and $\ell_\y = |y_1| = \ldots = |y_m|$.
  
  It remains to prove that by setting $C := 2n \gamma_4$ we have
  \begin{align} \label{eq:toshow}
\min_{A \in \algn_{n,m}} \Val{A} \le \delta(x,y) - C \le \min_{A \in \strc_{n,m}} \Val{A}.
  \end{align}
  
  We first prove the following three useful observations. Here for a string $z$ and indices $a\le b$ we denote the substring from $z[a]$ to $z[b]$ by $z[a..b]$.

\begin{claim} \label{cla:lcsobs}
  Let $x$ and $z_1,\ldots,z_k$ be binary strings. Set $z = z_1 \ldots z_n$. Then we have
  $$ \dLCS(x,z) = \min_{x(z_1),\ldots,x(z_k)} \sum_{j=1}^k \dLCS(x(z_j),z_j),  $$
  where $x(z_1),\ldots,x(z_k)$ range over all \emph{ordered partitions} of $x$ into $k$ substrings, i.e., $x(z_1) = x[i_0+1..i_1], x(z_2) = x[i_1+1..i_2], \ldots, x(z_k) = x[i_{k-1}+1..i_k]$ for any $0 = i_0 \le i_1 \le \ldots \le i_{k} = |x|$.
\end{claim}
\begin{proof}
  For any ordered partition, the substrings $x(z_j)$ are disjoint and ordered along $x$, so we can take the longest common subsequences of $(x(z_j),z_j)$, $j \in [k]$, and concatenate them to form a common subsequence of $(x,z)$. This shows $|\LCS(x,z)| \ge \sum_{j=1}^k |\LCS(x(z_j),z_j)|$. Since furthermore $|x| = \sum_{j=1}^k |x(z_j)|$ and $|z| = \sum_{j=1}^k |z_j|$ we obtain $\dLCS(x,z) \le \sum_{j=1}^k \dLCS(x(z_j),z_j)$.
  
  Now consider a longest common subsequence $s$ of $(x,y)$, which can be seen as a matching of symbols in $x$ and $y$. Let $J_j$ be the indices in $x$ that are matched to symbols in $z_j$ by $s$. Note that $\sum_{j=1}^k |J_j| = |\LCS(x,y)|$, as any matched symbol in $x$ is matched to some $z_j$. Also note that the matching is ordered, meaning that for any $i \in J_j$ and $i' \in J_{j'}$ with $j < j'$ we have $i < i'$. This allows to find an ordered partition $x(z_1),\ldots,x(z_k)$ of $x$ such that $x(z_j)$ contains the indices $J_j$ for any $j$. Finally, for this partition we have $\LCS(x(z_j),z_j) \ge |J_j|$ so that $\dLCS(x(z_j),z_j) \le |x(z_j)| + |z_j| - 2|J_j|$. Summing up over $j$, we obtain 
  $\sum_{j=1}^k \dLCS(x(z_j),z_j) \le |x| + |z| - 2|\LCS(x,z)| = \dLCS(x,z)$.
  Together both halves of the proof imply the desired statement.
\end{proof}

  \begin{claim} \label{cla:lcsgreedy}
    Let $z,w$ be binary strings and $\ell,k \in \mathbb{N}_0$. Then we have (1) $\dLCS(1^k z, 1^k w) = \dLCS(z,w)$ and (2) $\dLCS(0^\ell z, 1^k w) \ge \min\{k, \dLCS(z,1^k w)\}$. Symmetrically, we have (2') $\dLCS(0^k z, 1^\ell w) \ge \min\{k, \dLCS(0^k z, w)\}$, and we obtain more symmetric statements by reversing all involved strings.
  \end{claim}
  \begin{proof}
    (1) It suffices to show the claim for $k=1$, then the general statement follows by induction. Consider a \LCS\ $s$ of $(1z, 1w)$. At least one '1' is matched in $s$, as otherwise we can extend $s$ by matching both '1's. If exactly one '1' is matched in $s$, then the other '1' is free, so we may instead match the two '1's. Thus, without loss of generality a \LCS\ of $(1z,1w)$ matches the two '1's. This yields $|\LCS(1z,1w)| = 1 + |\LCS(z,w)|$. Hence, $\dLCS(1 z, 1 w) = |1 z| + |1 w| - 2|\LCS(1 z, 1 w)| = |z| + |w| - 2 |\LCS(z,w)| = \dLCS(z,w)$.
    
    (2) Fix a \LCS\ $s$ of $(0^\ell z, 1^k w)$. If $s$ starts with a 0, then it does not contain the leading $1^k$ of the second argument, leaving at least $k$ symbols unmatched, so that $\dLCS(0^\ell z, 1^k w) \ge k$. Otherwise, if $s$ starts with a 1, then it does not contain the leading $0^\ell$ of the first argument, so that $|\LCS(0^\ell z, 1^k w)| = |\LCS(z, 1^k w)|$. Then we have $\dLCS(0^\ell z, 1^k w) = |0^\ell z| + |1^k w| - 2|\LCS(0^\ell z, 1^k w)| \ge |z| + |1^k w| - 2|\LCS(z, 1^k w)| = \dLCS(z,1^k w)$.
  \end{proof}
  
  \begin{claim} \label{cla:prefix}
    Let $\ell \ge 0$. For any prefix $x'$ of $x$ we have $\dLCS(x',0^\ell) \ge \ell$. Moreover, if $x'$ is of the form $\guard(x_1) 0^{\gamma_3} \ldots \guard(x_i) 0^{\gamma_3}$ for some $0 \le i < n$ and $\ell \ge i \cdot (2\gamma_2+s_\x)$, then $\dLCS(x',0^\ell) = \ell$. Symmetric statements hold for any suffix of $x$.
  \end{claim}
  \begin{proof}
    We first show that for any $i \in [n]$ the string $\guard(x_i) 0^{\gamma_3}$ contains as many ones as zeroes, and any prefix of $\guard(x_i) 0^{\gamma_3}$ contains at least as many ones as zeroes.
    To this end, note that each $x_i$ has length $\ell_\x$ and contains $s_\x$ ones, so that the number of ones of $\guard(x_i) 0^{\gamma_3}$ is $2\gamma_2 + s_\x$, while the number of zeroes is $\ell_\x - s_\x + 2\gamma_1 + \gamma_3$, and we chose $\gamma_3$ such that both values are equal. For a prefix, note that $\guard(x_i)$ starts with $\gamma_2$ ones. Since each $\guard(x_i)$ contains $2\gamma_1 + \ell_\x - s_\x \le \gamma_2$ zeroes, any prefix of $\guard(x_i)$ has as most as many zeroes as ones. Thus, we would have to advance to $0^{\gamma_3}$ to see more zeroes than ones, however, even $\guard(x_i) 0^{\gamma_3}$ does not contain more zeroes than ones.
    
    Hence, any prefix $x'$ of $x$ contains at least as many ones as zeroes, implying $|\LCS(x',0^\ell)| \le |x'|/2$. This yields $\dLCS(x',0^\ell) = |x'| + |0^\ell| - 2|\LCS(x',0^\ell)| \ge \ell$. If $x'$ is of the form $\guard(x_1) 0^{\gamma_3} \ldots \guard(x_i) 0^{\gamma_3}$ and sufficiently many zeroes are available in $0^\ell$ then we have equality.
  \end{proof}

  Let us give names to the substrings consisting only of zeroes in $x$ and $y$. In $x$, we denote the $0^{\gamma_3}$-block after $\guard(x_i)$ by $Z^\x_i$, $i \in [n-1]$. In $y$, we denote the $0^{\gamma_3}$-block after $\guard(y_j)$ by $Z^\y_j$, $j \in [m-1]$. Moreover, we denote the prefix $0^{n \gamma_4}$ by $L^\y$ and the suffix $0^{n\gamma_4}$ by $R^\y$. 

\begin{figure}
\includegraphics[width=\textwidth]{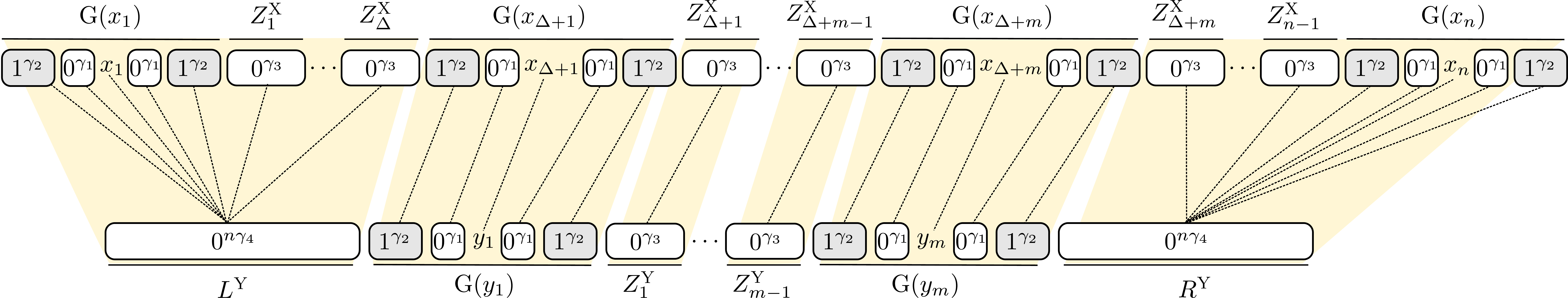}
\caption{Optimal traversal corresponding to structured alignment $A=\{(\Delta+j,j) \mid j\in [m]\} \in \strc_{n,m}$.}
\label{fig:lcs}
\end{figure}
  
  We now show the upper bound of \eqref{eq:toshow}, i.e., $\dLCS(x,y) \le 2n \gamma_4 + \min_{A \in \strc_{n,m}} \Val{A}$.
  Consider a structured alignment $A = \{(\Delta+1,1),\ldots,(\Delta+m,m)\} \in \strc_{n,m}$. 
  We construct an ordered partition of $x$ as in \claref{lcsobs} by setting (see \figref{lcs})
  \begin{align*}
    x(\guard(y_j)) &:= \guard(x_{\Delta+j}) \qquad \text{for }j \in [m],  \\
    x(Z^\y_j) &:= Z^\x_{\Delta+j} \qquad \text{for }j \in [m-1],  \\
    x(L^\y) &:= \guard(x_1) Z^\x_1 \ldots \guard(x_{\Delta}) Z^\x_{\Delta},  \\
    x(R^\y) &:= Z^\x_{\Delta+m} \guard(x_{\Delta+m+1}) \ldots Z^\x_{n-1} \guard(x_n). 
  \end{align*} 
  Note that indeed these strings partition $x$ and $y$, respectively. Thus, \claref{lcsobs} yields 
  $$ \dLCS(x,y) \le \dLCS(x(L^\y),L^\y) + \dLCS(x(R^\y),R^\y) + \sum_{j=1}^m \dLCS(\guard(x_{\Delta+j}),\guard(y_j)) + \sum_{j=1}^{m-1} \dLCS(Z^\x_{\Delta+j},Z^\y_j).  $$
  Since $L^\y = 0^{n\gamma_4}$ and $x(L^\y)$ is a prefix of $x$ of the correct form, by \claref{prefix} we have $\dLCS(x(L^\y),L^\y) = n\gamma_4$ (note that $\gamma_4$ is chosen sufficiently large to make \claref{prefix} applicable). Similarly we obtain $\dLCS(x(R^\y),R^\y) = n\gamma_4$. Since $Z_i^\x = Z_j^\y = 0^{\gamma_3}$ we have $\dLCS(Z^\x_{\Delta+j},Z^\y_j) = 0$. Finally, by matching the guarding ones and zeroes of $\guard(x_{\Delta+j}) = 1^{\gamma_2} 0^{\gamma_1} x_{\Delta+j} 0^{\gamma_1} 1^{\gamma_2}$ and $\guard(y_j) = 1^{\gamma_2} 0^{\gamma_1} y_j 0^{\gamma_1} 1^{\gamma_2}$ we obtain $\dLCS(\guard(x_{\Delta+j}),\guard(y_j)) \le \dLCS(x_{\Delta+j},y_j)$. Hence, we have
  $$ \dLCS(x,y) \le 2n\gamma_4 + \sum_{(i,j) \in A} \dLCS(x_i,y_j). $$
  As $A \in \strc_{n,m}$ was arbitrary, we proved $\dLCS(x,y) \le 2n \gamma_4 + \min_{A \in \strc_{n,m}} \Val{A}$, as desired.
  
  It remains to prove the lower bound of \eqref{eq:toshow}, i.e., $\dLCS(x,y) \ge 2n\gamma_4 + \min_{A \in \algn_{n,m}} \Val{A}$.
  As in \claref{lcsobs}, let $x(L^\y)$, $x(\guard(y_j))$ for $j \in [m]$, $x(Z^\y_j)$ for $j \in [m-1]$, $x(R^\y)$ be an ordered partition of $x$ such that 
  $$ \dLCS(x,y) = \dLCS(x(L^\y),L^\y) + \dLCS(x(R^\y),R^\y) + \sum_{j=1}^m \dLCS(x(\guard(y_j)),\guard(y_j)) + \sum_{j=1}^{m-1} \dLCS(x(Z^\y_j),Z^\y_j). $$
  Clearly, we can bound $\dLCS(x(Z^\y_j),Z^\y_j) \ge 0$. Since $L^\y = 0^{n\gamma_4}$ and $x(L^\y)$ is a prefix of $x$, by \claref{prefix} we have $\dLCS(x(L^\y),L^\y) \ge n\gamma_4$, and similarly we get $\dLCS(x(R^\y),R^\y) \ge n\gamma_4$.
  It remains to construct an alignment $A \in \algn_{n,m}$ satisfying 
  \begin{align} \label{eq:Alcs}
    \Val{A} \le \sum_{j=1}^m \dLCS(x(\guard(y_j)),\guard(y_j)), 
  \end{align}
  then together we have shown the desired inequality 
  $\dLCS(x,y) \ge 2n\gamma_4 + \min_{A \in \algn_{n,m}} \Val{A}$.
  
  Let us construct such an alignment $A$. For any $j \in [m]$, if $x(\guard(y_j))$ contains more than half of some $x_{i'}$ (which is part of $\guard(x_{i'})$), then let $i$ be the leftmost such index and align $i$ and $j$. Note that the set $A$ of all these aligned pairs $(i,j)$ is a valid alignment in $\algn_{n,m}$, since no $x_i$ or $y_j$ can be aligned more than once.
  
  Since by definition we have $\Val{A} = \sum_{(i,j) \in A} \dLCS(x_i,y_j) + (m-|A|) \max_{i,j} \dLCS(x_i,y_j)$ and since $\max_{i,j} \dLCS(x_i,y_j) \le \max_{i,j}(|x_i|+|y_j|) = \ell_\x+\ell_\y$, in order to show \eqref{eq:Alcs} it suffices to prove the following two claims.
  
  \begin{claim}  
    For any aligned pair $(i,j) \in A$ we have $\dLCS(x(\guard(y_j)),\guard(y_j)) \ge \dLCS(x_i,y_j)$.
  \end{claim}
  \begin{proof}
    Recall that $x(\guard(y_j))$ contains more than half of $x_i$. 
    First consider the case that $x(\guard(y_j))$ touches not only $\guard(x_i)$ but also $\guard(x_{i'})$ for some $i' \ne i$. As between $x_i$ and $\guard(x_{i'})$ there is at least one block of zeroes $0^{\gamma_3}$ and half of the guarding of $\guard(x_i)$ (i.e., $1^{\gamma_2} 0^{\gamma_1}$ or $0^{\gamma_1} 1^{\gamma_2}$), we obtain $|x(\guard(y_j))| \ge |0^{\gamma_3}| + |1^{\gamma_2} 0^{\gamma_1}| = \gamma_3 + \gamma_2 + \gamma_1$. 
    Thus, any matching of $x(\guard(y_j))$ and $\guard(y_j)$ leaves at least $|x(\guard(y_j))| - |\guard(y_j)| \ge (\gamma_3 + \gamma_2 + \gamma_1) - (2\gamma_2+2\gamma_1 + \ell_\y) = \gamma_3 - \gamma_2 - \gamma_1 - \ell_\y \ge \ell_\x+\ell_\y$ unmatched symbols, implying $\dLCS(x(\guard(y_j)),\guard(y_j)) \ge \ell_\x+\ell_\y \ge \dLCS(x_i,y_j)$.
    
    Now consider the remaining case, where $x(\guard(y_j))$ touches no other $\guard(x_{i'})$. In this case, $x(\guard(y_j))$ is a substring of $0^{\gamma_3} \guard(x_i) 0^{\gamma_3}$, i.e., we can write $x(\guard(y_j))$ as $0^{h_L} z 0^{h_R}$, where $z$ is a substring of $\guard(x_i)$. Since $\guard(y_j)$ starts with $\gamma_2$ ones, by \claref{lcsgreedy}.(2) we have $\dLCS(x(\guard(y_j)), \guard(y_j)) \ge \min\{\gamma_2, \dLCS(z 0^{h_R}, \guard(y_j))\}$. Since $\gamma_2 \ge \ell_\x + \ell_\y \ge \dLCS(x_i,y_j)$, it suffices to bound $\dLCS(z 0^{h_R}, \guard(y_j))$ from below. By a symmetric argument, we eliminate the block $0^{h_R}$ and only have to bound $\dLCS(z, \guard(y_j))$ from below. 
    We can assume that $|z| > |\guard(y_j)| - \gamma_2$, since otherwise $\dLCS(z,\guard(y_j)) \ge \gamma_2 \ge \ell_\x + \ell_\y \ge \dLCS(x_i,y_j)$. 
    Thus, we have $\guard(y_j) = 1^{\gamma_2} 0^{\gamma_1} y_j 0^{\gamma_1} 1^{\gamma_2}$ and can write $z$ as $1^{r_L} 0^{\gamma_1} x_i 0^{\gamma_1} 1^{r_R}$ with $r_L, r_R > 0$.  
    By \claref{lcsgreedy}.(1) we have $\dLCS(z,\guard(y_j)) = \dLCS(0^{\gamma_1} x_i 0^{\gamma_1} 1^{r_R}, 1^{\gamma_2 - r_L} 0^{\gamma_1} y_j 0^{\gamma_1} 1^{\gamma_2})$. By \claref{lcsgreedy}.(2'), this yields $\dLCS(z,\guard(y_j)) \ge \min\{\gamma_1, \dLCS(0^{\gamma_1} x_i 0^{\gamma_1} 1^{r_R}, 0^{\gamma_1} y_j 0^{\gamma_1} 1^{\gamma_2})\}$, and since $\gamma_1 \ge \ell_\x+\ell_\y \ge \dLCS(x_i,y_j)$ it suffices to bound the latter term. By a symmetric argument we eliminate the ones on the right side, and it suffices to bound $\dLCS(0^{\gamma_1} x_i 0^{\gamma_1}, 0^{\gamma_1} y_j 0^{\gamma_1})$. Using \claref{lcsgreedy}.(1) twice, this is equal to $\dLCS(x_i,y_j)$. Hence, we have shown the desired inequality $\dLCS(x(\guard(y_j)),\guard(y_j)) \ge \dLCS(x_i,y_j)$.
  \end{proof}   
  \begin{claim}
    If $j$ is unaligned in $A$, then $\dLCS(x(\guard(y_j)),\guard(y_j)) \ge \ell_\x+\ell_\y$.
  \end{claim}
  \begin{proof}
    Since $x(\guard(y_j))$ contains less than half of any $x_{i}$, examining the structure of $x$ we see that $x(\guard(y_j))$ is a substring\footnote{Actually $x(\guard(y_j))$ could also be a substring  of $1^{\gamma_2} 0^{\gamma_1} x_1$ or of $x_n 0^{\gamma_1} 1^{\gamma_2}$. We treat these border cases by setting $x_0 := x_1$ and $x_{n+1} := x_n$ and letting from now on $0 \le i \le n$.} of $P := x_i 0^{\gamma_1} 1^{\gamma_2} 0^{\gamma_3} 1^{\gamma_2} 0^{\gamma_1} x_{i+1}$ for some $1 \le i < n$, where at most half of $x_i$ and $x_{i+1}$ can be part of $x(\guard(y_j))$. If $x(\guard(y_j))$ contains ones to the left and to the right of $0^{\gamma_3}$ in $P$, then $x(\guard(y_j))$ contains at least $\gamma_3$ zeroes. 
    Since $\guard(y_j)$ contains $2\gamma_1 + \ell_\y - s_\y \le 2\gamma_1 + \ell_\y$ zeroes, at most $2\gamma_1+\ell_\y$ zeroes of $x(\guard(y_j))$ can be matched, leaving at least $\gamma_3 - 2\gamma_1 - \ell_\y$ unmatched zeroes. Thus,  $\dLCS(x(\guard(y_j)),\guard(y_j)) \ge \gamma_3-2\gamma_1-\ell_\y \ge \ell_\x+\ell_\y$. 
    Otherwise, if $x(\guard(y_j))$ contains only ones to the left of $0^{\gamma_3}$ in $P$ (or only to the right), then $x(\guard(y_j))$ contains at most $\gamma_2+\ell_\x$ ones. Thus, among the $2\gamma_2 + s_\y \ge 2\gamma_2$ ones of $\guard(y_j)$ at least $\gamma_2 -\ell_\x$ ones remain unmatched, implying $\dLCS(x(\guard(y_j)),\guard(y_j)) \ge \gamma_2 -\ell_\x \ge \ell_\x + \ell_\y$.
  \end{proof}   

This finishes the proof of \lemref{lcscorrect}.
\end{proof}

\section{Edit Distance} \label{sec:edit}

We first show that the trivial cases of \EDITallc\ can be solved in constant time.  For all other cases, on binary strings we present a reduction from \EDITallc\ to $\EDIT(\csubst')$ and vice versa, see \secref{EDITreductions}.
Then in \secref{EDIThardness} we prove a conditional lower bound of $\Oh(m^{2-\eps})$ for \EDITc\ by applying our \aligngad-framework.
Finally, in \secref{EDITalgo} we show that \EDITallc\ can be solved in time $\tilde \Oh(n + m^2)$, which matches our lower bound.

\subsection{Easy Reductions}
\label{sec:EDITreductions}

All of our reductions are of the following form.
Let $E_1 = \EDITallc$ and $E_2 = \EDIT(\cdelx', \cdely', \cmatch', \csubst')$ be two variants of the edit distance and denote the cost of any traversal $T$ with respect to $E_i$ by $\delta_{E_i}(T)$. We say that $E_1$ and $E_2$ are \emph{equivalent}, if there are constants $\alpha, \beta$ such that for any traversal $T$ we have $\delta_{E_1}(T) = \alpha \cdot \delta_{E_2}(T) + \beta$. Then the complexity of computing $E_1$ and $E_2$ is asymptotically equal. 

\begin{lem} \label{lem:editreduction}
  (1) \EDITallc\ can be solved in constant time if $\csubst = \cmatch$ or $\cdelx + \cdely \le \min\{\cmatch,\csubst\}$. 
  Otherwise, \EDITallc\ on binary strings is equivalent to $\EDIT(\csubst')$ on binary strings for some $0 < \csubst' \le 2$.
  
  (2) \EDITallc\ is equivalent to $\EDIT(\cdelx',\cdely',\cmatch',\csubst')$ for some positive integers $\cdelx',\cdely',\cmatch',\csubst'$.
\end{lem}
Note that by the first statement, hardness for general rational cost parameters follows by proving hardness of $\Edit(\csubst')$ for $0<\csubst' \le 2$. The second statement allows us to assume positive integer costs when giving an algorithm for \EDITallc\ in \secref{EDITalgo}.
\begin{proof}[Proof of Lemma~\ref{lem:editreduction}]
  Let $x,y$ be strings of length $n,m$. By symmetry, we may assume $n \ge m$. Observe that we can write the cost of any traversal $T$ with respect to \EDITallc\ as
  $$ \dEDIT(T) = A\cdot \cmatch + B\cdot \csubst + C\cdot (\cdelx + \cdely) + (n-m)\cdot \cdelx, $$
  for some $A,B,C \ge 0$ with $A+B+C = m$, since matchings and substitutions touch as many symbols in $x$ as in $y$, so that we need exactly $n-m$ more deletions in $x$ than deletions in $y$.
  
  (1) If $\cdelx + \cdely \le \min\{\cmatch,\csubst\}$, then we can replace any matching or substitution by a deletion in $x$ and a deletion in $y$ without increasing the cost. Thus, an optimal traversal has $C=m$ and minimal cost $n \cdot \cdelx + m \cdot \cdely$, which can be computed in constant time.
  Similarly, if $\cmatch = \csubst$, then the minimal cost is independent of the symbols in $x$ and $y$. We may arbitrarily set $A+B$ and $C$ subject to $A+B+C = m$ and $A+B, C \ge 0$, and the minimal cost is $m \cdot \min\{\cmatch, \cdelx+\cdely\} + (n-m) \cdelx$, which can be computed in constant time.
  
  Now assume that $\cmatch \ne \csubst$ and $\cdelx + \cdely > \min\{\cmatch,\csubst\}$. Restricting our attention to binary strings, by flipping all symbols in $y$ but not in $x$ we can swap the costs of matching and substitution. Thus, we may assume that $\csubst > \cmatch$ (and $\cdelx+\cdely > \cmatch$).
  We set
  \begin{align*}
    \csubst' := \alpha (\csubst - \cmatch) \quad \text{where} \quad \alpha := \tfrac{2}{\cdelx + \cdely - \cmatch}.
  \end{align*}
  One can easily verify that for any traversal $T$ with cost $\dEDIT(T) = A\cdot \cmatch + B\cdot \csubst + C\cdot (\cdelx + \cdely) + (n-m)\cdot \cdelx$ (with respect to \EDITallc) we have
  $$ \alpha \dEDIT(T) - \alpha m \cdot \cmatch + (n-m)(1-\alpha \cdelx) = B \cdot \csubst' + C \cdot 2 + (n-m).$$
  As the latter is the cost of $T$ with respect to $\EDIT(\csubst')$, this proves that $\EDIT(\csubst')$ is equivalent to \EDITallc. 
  Finally, note that $\csubst' > 0$. If $\csubst' > 2$, then we can replace it by~2 without changing the cost of the optimal traversal, since we can replace any substitution (of cost 2) by a deletion and an insertion (both of cost 1). This yields $0 < \csubst' \le 2$.
  
  (2) Since we always assume all operation costs to be rationals, without loss of generality $\cdelx,\cdely,\cmatch,\csubst$ have a common denominator~$D$. We obtain positive integral operation costs by setting $\cdelx' := D \cdelx + M$, $\cdely' := D \cdely + M$, $\cmatch' := D \cmatch + 2M$, $\csubst' := D \csubst + 2M$ for a sufficiently large integer $M$. Both variants are equivalent, since $\dEDIT(T)$ is changed to 
  $$ D \dEDIT(T) + m \cdot 2M + (n-m) \cdot M. \qedhere $$
\end{proof}

\subsection{Hardness Proof}
\label{sec:EDIThardness}

In this section we study the edit distance with matching cost 0, deletion and insertion cost 1, and substitution cost $0 < \csubst \le 2$. We abbreviate $\dEDIT = \dEDITc$.

\begin{lem}
$\EDIT(\csubst)$ admits coordinate values by setting
$$ \oleft := 11100,\; \zleft := 10011,\; \oright := 00111,\; \zright := 11001. $$
\end{lem}
\begin{proof}
  All four values have the same length and the same number of ones, so they have equal type. Using \facref{EDITgreedy}.(1) (to be proven below), we have $\dEDIT(\zleft,\zright) = \dEDIT(10011,11001) = \dEDIT(0011,1001) = \dEDIT(001,100)$. Depending on $\csubst$, the optimal traversal of $(001,100)$ is either to delete both ones or to substitute the first and last symbols. This yields $\dEDIT(001,100) = \min\{2,2\csubst\}$. Similarly, we obtain $\dEDIT(\oleft,\zright) = \dEDIT(\zleft,\oright) = \dEDIT(\zleft,\zright) = \min\{2,2\csubst\}$ and $\dEDIT(\oleft,\oright) = \dEDIT(11100,00111) = \min\{4,4\csubst\}$. Hence, $\dEDIT(\oleft,\oright) > \dEDIT(\oleft,\zright) = \dEDIT(\zleft,\oright) = \dEDIT(\zleft,\zright)$.
  
\end{proof}

\begin{defn} \label{def:EDIT}
Consider instances $x_1,\ldots,x_n \in \inputs_{t_\x}$ and $y_1,\ldots,y_m \in \inputs_{t_\y}$ with $n \ge m$ and types $t_\x = (\ell_\x,s_\x), t_\y = (\ell_\y,s_\y)$. 
We define the parameters $\rho := 2 \lceil 1/\csubst \rceil$, $\gamma_1 := 10 \rho(\ell_\x+\ell_\y)$, $\gamma_2 := 6 \rho \gamma_1 + 5 s_\x - \ell_\x$, and $\gamma_3 := 2 \gamma_2$ (since $\csubst$ is constant, these parameters are $\Theta(\ell_\x + \ell_\y)$).

To \emph{guard} a string by blocks of zeroes and ones, we set $\guard(z) := (1^{\gamma_1} 0^{\gamma_1})^\rho z (0^{\gamma_1} 1^{\gamma_1})^\rho$. 
Now the \aligngad\ gagdet is
\begin{align*}
x &:=\qquad \, \guard(x_1) \; 0^{\gamma_2} \; \guard(x_2) \; 0^{\gamma_2} \,\ldots\, \guard(x_{n-1}) \; 0^{\gamma_2} \; \guard(x_n),  \\
y &:= \, 0^{n\gamma_3} \; \guard(y_1) \; 0^{\gamma_2} \; \guard(y_2) \; 0^{\gamma_2} \,\ldots\, \guard(y_{m-1}) \; 0^{\gamma_2} \; \guard(y_m) \; 0^{n\gamma_3}.
\end{align*}
\end{defn}

Let us provide some intuition on the complex guarding $\guard(z)$, which contains more parts compared to the construction for LCS.
Consider a block $B = (1^\gamma 0^\gamma)^\rho$. Clearly, $B$ can be completely matched to $B$, resulting in a cost of 0. Consider a slight perturbation $B'$ of $B$ by prepending $\Delta$ ones and deleting the last $\Delta$ zeroes. Then the edit distance of $B$ and $B'$ is at most $2 \Delta$, since we may delete the prepended ones in $B'$ and the additional zeroes at the end of $B$. Another upper bound for the edit distance of $B$ and $B'$ is $2 \rho \cdot \Delta \csubst$, since we may match the first $\gamma$ ones, then substitute the next $\Delta$ symbols, then match the next $\gamma - \Delta$ zeroes, and so on. By choosing $\rho := 2 \lceil 1/\csubst \rceil$, the traversal using substitutions is more expensive, and indeed we prove that then the edit distance is at least $2 \Delta$. This provides a building block where we got rid of substitutions and where slight perturbations are severely punished. Thus, our guarding $\guard(z) = (1^{\gamma_1} 0^{\gamma_1})^\rho z (0^{\gamma_1} 1^{\gamma_1})^\rho$ ensures that an optimal traversal of $\guard(x)$ and $\guard(y)$ aligns $x$ and $y$, and this also holds after small perturbations.

\begin{lem} \label{lem:EDITcorrect}
  For any $0 < \csubst \le 2$, \defref{EDIT} realizes an \aligngad\ gadget for $\EDIT(\csubst)$.
\end{lem}

Thus, \thmref{main} is applicable, implying a lower bound of $\Oh(m^{2-\eps})$ for $\EDIT(\csubst)$. Combining this with \lemref{editreduction} proves \thmref{EDIT}. 
We remark that our construction is no \emph{unbalanced} \aligngad\ gadget, as the length of $y$ grows linearly in $n$, not necessarily in $m$. Thus, we do not obtain a conditional lower bound of $\Oh((n m)^{1-\eps})$ (i.e., not for $m \approx n^\alpha$ for all $0 < \alpha < 1$), which in fact is ruled out by the algorithmic result of \thmref{EDITalgo}, see \secref{EDITalgo}.

In the proof of \lemref{EDITcorrect} we make use of the following basic observations.

\begin{fact} \label{fac:EDITgreedy}
    Let $x,y,z$ be binary strings and $\ell,k \in \mathbb{N}_0$. Then we have (1) $\dEDIT(1^k x, 1^k y) = \dEDIT(x,y)$, (2) $\dEDIT(x,y) \ge \big| |x| - |y| \big|$ and (3) $|\dEDIT(x z,y) - \dEDIT(x,y)| \le |z|$. We obtain symmetric statements by replacing all 1's by 0's and by reversing all involved strings.
  \end{fact}
  \begin{proof}
    We show (1) for $k=1$, then the general statement follows by induction. 
    Consider an optimal traversal $T$ of $1x, 1y$. If both '1's are deleted in $T$, then we can instead match them and improve $T$, contradicting optimality. If exactly one '1' is matched or substituted, then the other '1' is deleted, so we may instead match the two '1's without increasing cost. Thus, without loss of generality an optimal traversal of $(1x,1y)$ matches the two '1's. 
    
    For (2), note that matchings and substitutions touch as many symbols in $x$ as in $y$. Hence, there have to be at least $|x| - |y|$ deletions in $x$ and at least $|y|-|x|$ deletions in $y$.
    
    For (3), taking an optimal traversal of $(x,y)$ and appending $|z|$ deletions of the symbols in $z$ shows that $\dEDIT(xz,y) \le \dEDIT(x,y) + |z|$. For the other direction, consider an optimal traversal~$T$ of $(xz,y)$. Replace any matching or substitution of a symbol in $z$ with a symbol $y[j]$ in $y$ by a deletion of $y[j]$. Also remove every deletion of a symbol in $z$. This results in a traversal $T'$ of $(x,y)$ with cost at most $\dEDIT(xz,y) + |z|$, as we introduced at most $|z|$ deletions in $y$. This proves the desired inequality $\dEDIT(x,y) \le \dEDIT(xz,y)+|z|$.
  \end{proof}

\begin{fact} \label{fac:EDITblock}
    Let $\ell, m, r \ge 0$. Then for any $x \in \{0^\ell 1^m 0^r, 1^{m-\ell-r}, 1^{m-\ell} 0^r, 0^\ell 1^{m-r}\}$ we have 
    $\dEDIT(x, 1^m) \ge |\ell-r| + \csubst \cdot \min\{\ell,r\}.$
\end{fact}
\begin{proof}
  \facref{EDITgreedy}.(2) yields $\dEDIT(0^\ell 1^m 0^r, 1^m), \dEDIT(1^{m-\ell-r},1^m) \ge \ell + r \ge |\ell-r| + \csubst \cdot \min\{\ell,r\}$, since $\csubst \le 2$. For $x = 0^\ell 1^{m-r}$, consider any optimal traversal $T$. If $T$ substitutes $s$ zeroes and deletes the remaining $\ell-s$ zeroes, then 
  $\dEDIT(0^\ell 1^{m-r}, 1^m) = \csubst \cdot s + (\ell-s) + \dEDIT(1^{m-r}, 1^{m-s})$. By \facref{EDITgreedy}.(1), $\dEDIT(1^{m-r}, 1^{m-s}) = \dEDIT(\epsilon, 1^{|r-s|}) = |r-s|$, where $\epsilon$ is the empty string. Hence, 
  $\dEDIT(0^\ell 1^{m-r}, 1^m) = \min_{0 \le s \le \ell}\{ \csubst \cdot s + \ell-s + |r-s| \}$. A short case analysis shows that this term is minimized for $s = \min\{\ell,r\}$, where it evaluates to $\csubst \cdot \min\{\ell,r\} + \ell + r - 2 \min\{\ell,r\} = \csubst \cdot \min\{\ell,r\} + |\ell-r|$.
  The case $x = 1^{m-\ell} 0^r$ is symmetric.
\end{proof}

For a string $y$ and indices $a\le b$ we denote the substring from $y[a]$ to $y[b]$ by $y[a..b]$.

\begin{fact} \label{fac:EDITobs}
  Let $x$ and $y_1,\ldots,y_k$ be binary strings. Set $y = y_1 \ldots y_n$. Then we have
  $$ \dEDIT(x,y) = \min_{x(y_1),\ldots,x(y_k)} \sum_{j=1}^k \dEDIT(x(y_j),y_j),  $$
  where $x(y_1),\ldots,x(y_k)$ ranges over all \emph{ordered partitions} of $x$ into $k$ substrings, i.e., $x(y_1) = x[i_0+1..i_1], x(y_2) = x[i_1+1..i_2], \ldots, x(y_k) = x[i_{k-1}+1..i_k]$ for any $0 = i_0 \le i_1 \le \ldots \le i_{k} = |x|$.
\end{fact}
\begin{proof}
  For any ordered partition, the substrings $x(y_j)$ are disjoint and ordered along $x$, so we can concatenate (optimal) traversals of $(x(y_j),y_j)$, $j \in [k]$, to form a traversal of $(x,y)$. This shows $\dEDIT(x,y) \le \sum_{j=1}^k \dEDIT(x(y_j),y_j)$.
  
  Now let $T$ be an optimal traversal of $(x,y)$. Let $J_j$ be the indices in $x$ that appear in a matching or substitution operation with symbols in $y_j$. Note that these sets are ordered, in the sense that for any $i \in J_j$ and $i' \in J_{j'}$ with $j < j'$ we have $i < i'$. This allows to find an ordered partition $x(y_1),\ldots,x(y_k)$ of $x$ such that $x(y_j)$ contains the indices $J_j$ for any $j$.
  Let us denote the total cost of the substitutions involving $y_j$ by $s_j$. 
  Since traversal $T$ deletes $|y_j| - |J_j|$ symbols in $y_j$ and $|x(y_j)| - |J_j|$ symbols in $x(y_j)$, we have $\delta(T) = \sum_{j=1}^k |y_j| + |x(y_j)| - 2|J_j| + s_j$. Clearly, we can construct a traversal of $(x(y_j),y_j)$ that follows the matchings and substitutions in $J_j$ and deletes all other symbols, showing $\dEDIT(x(y_j),y_j) \le |y_j| + |x(y_j)| - 2|J_j| + s_j$. By optimality of $T$, we obtain $\dEDIT(x,y) \ge \sum_{i=1}^k \dEDIT(x(y_j),y_j)$.
\end{proof}

\begin{proof}[Proof of \lemref{EDITcorrect}]
From now on let $x,y$ be as in \defref{EDIT}. Observe that indeed $x$ only depends on $m,t_\y$, and $x_1,\ldots,x_n$, and $\type(x)$ only depends on $n,m,t_\x$, and $t_\y$, and similarly for $y$. Moreover, $x$ and $y$ can clearly be constructed in time $\Oh((n+m)(\ell_\x+\ell_\y))$, where $\ell_\x = |x_1| = \ldots = |x_n|$ and $\ell_\y = |y_1| = \ldots = |y_m|$.
  
  It remains to prove that for some $C$, we have
  \begin{align} \label{eq:edittoshow}
\min_{A \in \algn_{n,m}} \Val{A} \le \delta(x,y) - C \le \min_{A \in \strc_{n,m}} \Val{A}.
  \end{align}
  We will set 
  $$C := 2 n \gamma_3 - \beta (n-m)(\gamma_4 + \gamma_2), $$ 
  where 
  $$\beta := 1-\csubst/5 \quad \text{and} \quad \gamma_4 := 4 \rho \gamma_1 + \ell_\x.$$ Note that $\gamma_4$ is the length of $\guard(x_i)$.

  Let us give names to the substrings consisting only of zeroes in $x$ and $y$. In $x$, we denote the $0^{\gamma_2}$-block after $\guard(x_i)$ by $Z^\x_i$, $i \in [n-1]$. In $y$, we denote the $0^{\gamma_2}$-block after $\guard(y_j)$ by $Z^\y_j$, $j \in [m-1]$. Moreover, we denote the prefix $0^{n \gamma_3}$ by $L^\y$ and the suffix $0^{n\gamma_3}$ by $R^\y$. 
  
  We first prove the crucial property that for any prefix $x'$ of $x$ the distance $\dEDIT(x',L^\y)$ is essentially $|L^\y| - \beta |x'| = n \gamma_3 - \beta |x'|$. This is due to a careful choice of the parameters $\gamma_1, \gamma_2, \rho$.

  \begin{claim} \label{cla:editprefix}
    For any prefix $x'$ of $x$ we have $\dEDIT(x',L^\y) \ge n \gamma_3 - \beta |x'|$, with equality if $x'$ is of the form $\guard(x_1) 0^{\gamma_2} \ldots \guard(x_i) 0^{\gamma_2}$ for any $0 \le i < n$. Symmetric statements hold for $\dEDIT(x'',R^\y)$ where $x''$ is any suffix of $x$.
  \end{claim}
  \begin{proof}
    The parameter $\gamma_3$ is chosen such that $|x'| \le |x| \le |L^\y|$: Indeed, $|x| \le n(4 \rho \gamma_1 + \ell_\x + \gamma_2) \le n \cdot 2 \gamma_2 \le n \gamma_3 = |L^\y|$. 
    Observe that all zeroes of $x'$ can be matched to zeroes of $L^\y$, while all ones of $x'$ have to be substituted. The remaining zeroes of $L^\y$ have to be deleted. Denoting the number of ones in $x'$ by $\ell$, we obtain $\dEDIT(x',L^\y) = \ell \cdot \csubst + (|L^\y| - |x'|)$. We will show $\ell \ge |x'|/5$, with equality if $x'$ has the special form as in the statement. In other words, the \emph{relative number of ones} $\ell/|x'|$ is at least $1/5$, with equality if $x'$ has the special form. This implies $\dEDIT(x',L^\y) \ge n \gamma_3 - \beta |x'|$, with equality if $x'$ has the special form.
    
    Note that each $x_i$ has length $\ell_\x$ and contains $s_\x$ ones, so that  $\guard(x_i) Z^\x_i = (1^{\gamma_1} 0^{\gamma_1})^\rho x_i (0^{\gamma_1} 1^{\gamma_1})^\rho 0^{\gamma_2}$ contains $2 \rho \gamma_1 + (\ell_\x - s_\x) + \gamma_2$ zeroes and $2 \rho \gamma_1 + s_\x$ ones. The parameter $\gamma_2$ is chosen so that the number of zeroes is four times the number of ones, 
    implying that the relative number of ones is $1/5$.
    Note that any prefix of $(1^{\gamma_1} 0^{\gamma_1})^\rho$ has relative number of ones at least $1/2$. Since $x_i 0^{\gamma_1}$ has less than $2 \gamma_1$ zeroes 
    and $|(1^{\gamma_1} 0^{\gamma_1})^\rho| \ge 2 \gamma_1$, any prefix of $(1^{\gamma_1} 0^{\gamma_1})^\rho x_i 0^{\gamma_1}$ has relative number of ones at least $1/4$. Since any prefix of $1^{\gamma_1} (0^{\gamma_1} 1^{\gamma_1})^{\rho-1}$ has relative number of ones at least $1/2$, any prefix of $(1^{\gamma_1} 0^{\gamma_1})^\rho x_i (0^{\gamma_1} 1^{\gamma_1})^\rho$ has relative number of ones at least $1/4$. The relative number of ones decreases by adding any prefix of $0^{\gamma_2}$, however, for the final string $(1^{\gamma_1} 0^{\gamma_1})^\rho x_i (0^{\gamma_1} 1^{\gamma_1})^\rho 0^{\gamma_2}$, we already argued that the relative number of ones is $1/5$. This shows that the relative number of ones of any prefix of $x$ is at least $1/5$.
  \end{proof}

\begin{figure}

\includegraphics[width=\textwidth]{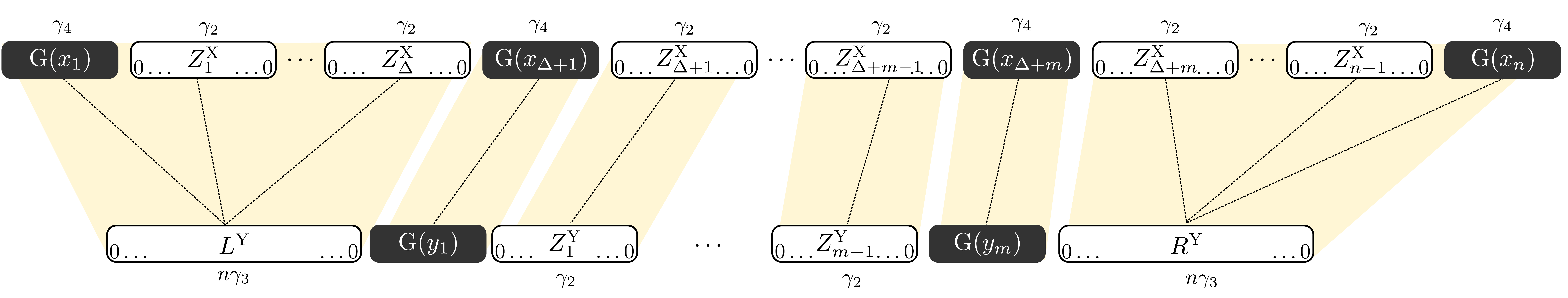}
\caption{Optimal traversal corresponding to structured alignment $A=\{(\Delta+j,j) \mid j\in [m]\} \in \strc_{n,m}$.}
\label{fig:edit}
\end{figure}
  
  We now show the upper bound of \eqref{eq:edittoshow}, i.e., $\dEDIT(x,y) \le C + \min_{A \in \strc_{n,m}} \sum_{(i,j) \in A} \dEDIT(x_i,y_j)$.
  Consider a structured alignment $A = \{(\Delta+1,1),\ldots,(\Delta+m,m)\} \in \strc_{n,m}$.
  We construct an ordered partition of $x$ as in \facref{EDITobs} by setting (see \figref{edit})
  \begin{align*}
    x(\guard(y_j)) &:= \guard(x_{\Delta+j}) \qquad \text{for }j \in [m],  \\
    x(Z^\y_j) &:= Z^\x_{\Delta+j} \qquad \text{for }j \in [m-1],  \\
    x(L^\y) &:= \guard(x_1) Z^\x_1 \ldots \guard(x_{\Delta}) Z^\x_{\Delta},  \\
    x(R^\y) &:= Z^\x_{\Delta+m} \guard(x_{\Delta+m+1}) \ldots Z^\x_{n-1} \guard(x_n). 
  \end{align*} 
  Note that indeed these strings partition $x$ and $y$, respectively. Thus, \facref{EDITobs} yields 
  $$ \dEDIT(x,y) \le \dEDIT(x(L^\y),L^\y) + \dEDIT(x(R^\y),R^\y) + \sum_{j=1}^m \dEDIT(\guard(x_{\Delta+j}),\guard(y_j)) + \sum_{j=1}^{m-1} \dEDIT(Z^\x_{\Delta+j},Z^\y_j).  $$
  Since $x(L^\y)$ is a prefix of $x$ of the correct form, by \claref{editprefix} we have $\dEDIT(x(L^\y),L^\y) = n\gamma_3 - \beta|x(L^\y)|$. 
  Symmetrically, we obtain $\dEDIT(x(R^\y),R^\y) = n\gamma_3 - \beta |x(R^\y)|$. Note that $|\guard(x_i) Z^\x_i| = \gamma_4 + \gamma_2$, so that $|x(L^\y)| + |x(R^\y)| = (n-m) (\gamma_4 + \gamma_2)$. Moreover, as $Z^\x_i = Z^\y_j = 0^{\gamma_2}$ we have $\dEDIT(Z^\x_{\Delta+j},Z^\y_j) = 0$. Finally, by matching all guarding zeroes and ones of $\guard(x_{\Delta+j}) = (1^{\gamma_1} 0^{\gamma_1})^\rho x_{\Delta+j} (0^{\gamma_1} 1^{\gamma_1})^\rho$ and $\guard(y_j) = (1^{\gamma_1} 0^{\gamma_1})^\rho y_j (0^{\gamma_1} 1^{\gamma_1})^\rho$ we conclude $\dEDIT(\guard(x_{\Delta+j}),\guard(y_j)) \le \dEDIT(x_{\Delta+j},y_j)$. This yields
  $$ \dEDIT(x,y) \le 2 n \gamma_3 - \beta (n-m)(\gamma_4 + \gamma_2) + \sum_{j=1}^m \dEDIT(x_{\Delta+j},y_j) = C + \sum_{(i,j) \in A} \dEDIT(x_i,y_j).  $$
  As $A \in \strc_{n,m}$ was arbitrary, the desired inequality follows.
  
  It remains to prove the lower bound of \eqref{eq:edittoshow}, i.e., $\dEDIT(x,y) \ge C + \min_{A \in \algn_{n,m}} \Val{A}$.
  As in \facref{EDITobs}, let $x(L^\y)$, $x(\guard(y_j))$ for $j \in [m]$, $x(Z^\y_j)$ for $j \in [m-1]$, $x(R^\y)$ be an ordered partition of $x$ such that 
  $$ \dEDIT(x,y) = \dEDIT(x(L^\y),L^\y) + \dEDIT(x(R^\y),R^\y) + \sum_{j=1}^m \dEDIT(x(\guard(y_j)),\guard(y_j)) + \sum_{j=1}^{m-1} \dEDIT(x(Z^\y_j),Z^\y_j). $$
  
  We define an alignment $A$ as follows. If there is some $i$ such that $x_i$ is contained in $x(\guard(y_j))$, then align $j$ with any such $i$. Otherwise leave $j$ unaligned.

  \begin{claim}\label{cla:dguards}
    We have 
    $$ \dEDIT(x(\guard(y_j)), \guard(y_j)) \ge \beta ( \gamma_4 - |x(\guard(y_j))| ) + \begin{cases} \dEDIT(x_i,y_j) & \text{if $j$ is aligned to $i$},  \\ \max_{i,j'} \dEDIT(x_i,y_{j'}) & \text{if $j$ is unaligned.} \end{cases} $$
  \end{claim}
  \begin{proof}
    If $|x(\guard(y_j))| \ge \gamma_2 $, then $|x(\guard(y_j))| \ge \gamma_2 \ge \gamma_4 + 2(\ell_\x + \ell_\y) \ge \gamma_4 + 2 \max_{i,j'} \dEDIT(x_i,y_{j'})$ 
    and by $\beta > 1/2$ the right hand side of the claim is at most 0, so the claim holds trivially. 
    Otherwise $x(\guard(y_j))$ is shorter than any $Z^\x_i = 0^{\gamma_2}$, implying that $x(\guard(y_j))$ is a substring of $0^{\gamma_2} \guard(x_i) 0^{\gamma_2}$ for some $i \in [n]$.
    
    We write $\guard(y_j)$ as $z_{-2\rho} \, z_{-2\rho+1} \, \ldots\, z_{2\rho-1} \, z_{2\rho}$, where $z_{-2k} = z_{2k} = 1^{\gamma_1}$, $z_{-2k+1} = z_{2k-1} = 0^{\gamma_1}$, and $z_0 = y_j$ (for all $1 \le k \le \rho$). As in \facref{EDITobs}, we split up $x(\guard(y_j))$ into $x(z_k)$, $-2\rho \le k \le 2\rho$, such that $\dEDIT(x(\guard(y_j)), \guard(y_j)) = \sum_{k=-2\rho}^{2\rho} \dEDIT(x(z_k),z_k)$. Similarly, we write $\guard(x_i)$ as $w_{-2\rho} \, w_{-2\rho+1} \, \ldots\, w_{2\rho-1} \, w_{2\rho}$. We denote the distance of the start of $x(z_k)$ to the start of $w_k$ by $\Delta_L(k)$, i.e., if $x(z_k) = x[a..b]$ and $w_k = x[c..d]$ we set $\Delta_L(k) := |a-c|$.  Similarly, we set $\Delta_R(k) := |b-d|$. For an illustration, see \figref{guards}. Note that $\Delta_R(k) = \Delta_L(k+1)$ holds for any $k$. 

\begin{figure}
\centering
\includegraphics[width=\textwidth]{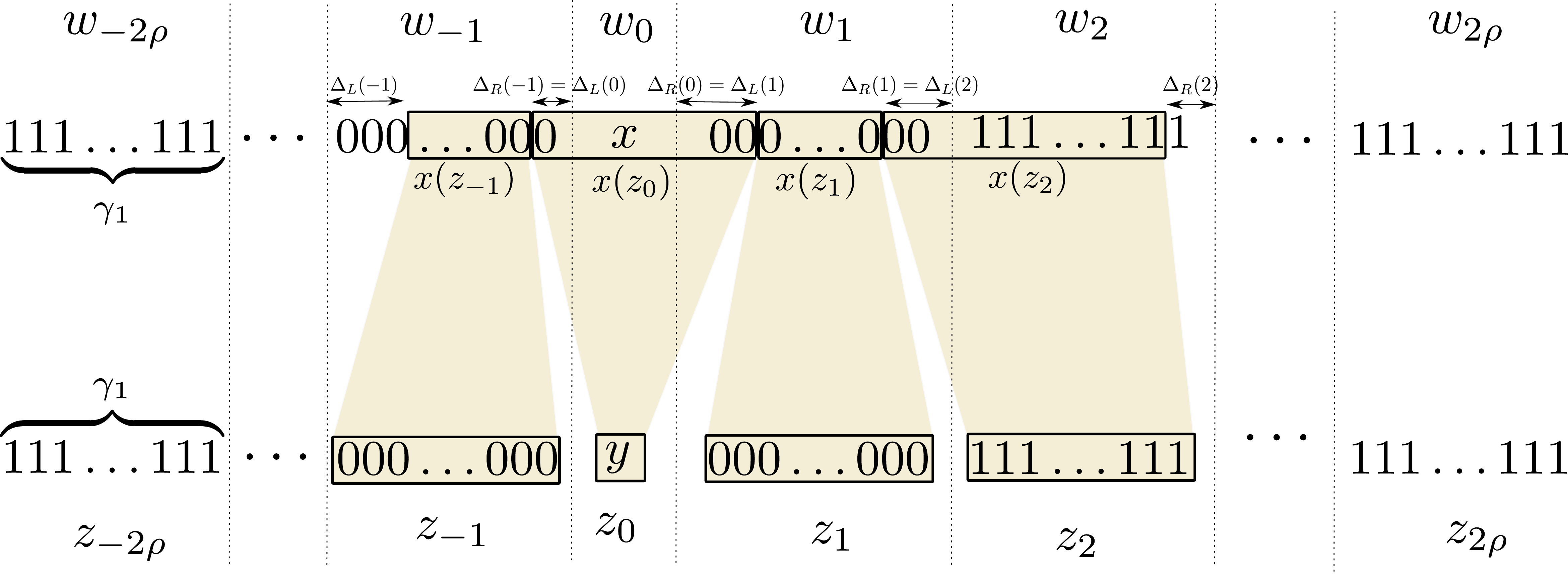}
\caption{Illustration for the proof of \claref{dguards}.}
\label{fig:guards}
\end{figure}
    
    First assume $(*)$: for some $k \ne 0$ the string $x(z_k)$ is longer than $\frac 54 \gamma_1$ or $x(z_k)$ has less than $\frac 34 \gamma_1$ common symbols with $z_k$. Then clearly $\dEDIT(x(\guard(y_j)), \guard(y_j)) \ge \dEDIT(x(z_k),z_k) \ge \frac 14 \gamma_1$. By \facref{EDITgreedy}.(2), we also have $\dEDIT(x(\guard(y_j)), \guard(y_j)) \ge |\guard(y_i)| - |x(\guard(y_j))| = \gamma_4 - |x(\guard(y_j))|$. As a linear combination of these two lower bounds, we obtain $\dEDIT(x(\guard(y_j)), \guard(y_j)) \ge \beta ( \gamma_4 - |x(\guard(y_j))| ) + (1-\beta) \frac 14 \gamma_1$. Since $(1-\beta) \frac 14 \gamma_1 = \frac \csubst{20} \gamma_1 \ge \ell_\x + \ell_\y \ge \max_{i,j'} \dEDIT(x_i,y_{j'})$, 
    we have proven the statement in this case.
    
    If (*) does not hold, then we have $\Delta_L(k), \Delta_R(k) \le \frac 12 \gamma_1$ for any $|k| > 1$: It suffices to show the claim for any even $k \ne 0$, since $\Delta_R(k) = \Delta_L(k+1)$. For any even $k \ne 0$, the string $x(z_k)$ has to contain the majority of a block $w_{\ell}$ with even $\ell \ne 0$. Since the numbers of blocks are identical in $\guard(y_j)$ and $\guard(x_i)$, $x(z_k)$ has to contain the majority of $w_k$ for any even $k \ne 0$. Specifically, $x(z_k)$ contains at least $\frac 34 \gamma_1$ symbols of $w_k$ and has length at most $\frac 54 \gamma_1$, implying the desired inequalities for $\Delta_L(k), \Delta_R(k)$. Note that in this case $i$ and $j$ are aligned.
    
    Note that for even $k \ne 0$ we obtain $x(z_k)$ from $w_k = z_k = 1^{\gamma_1}$ by either deleting a prefix of $\Delta_L(k)$ ones or prepending $\Delta_L(k)$ zeroes, and by either deleting a suffix of $\Delta_R(k)$ ones or by appending $\Delta_R(k)$ zeroes. Hence, \facref{EDITblock} shows that 
    \begin{align}
       \dEDIT(x(z_k),z_k) \ge |\Delta_L(k) - \Delta_R(k)| + \csubst \cdot \min\{\Delta_L(k), \Delta_R(k)\}. \label{eq:xzk}
    \end{align}
    The same argument works for any $k$ with $|k|>1$. For $k \in \{-1,1\}$ the argument does not work, since $z_{-1} = z_1 = 0^{\gamma_1}$ is not sorrounded by blocks of $1^{\gamma_1}$. However, for $k \in \{-1,1\}$ we have the weaker $\dEDIT(x(z_k),z_k) \ge |\Delta_L(k) - \Delta_R(k)|$ by \facref{EDITgreedy}.(2).
    Moreover, by \facref{EDITgreedy}.(3) we have 
    \begin{align}
      \dEDIT(x(z_0),z_0) \ge \dEDIT(x_i,y_j) - \Delta_L(0) - \Delta_R(0). \label{eq:xzO}
    \end{align} 
    
    Combining these inequalities yields $\dEDIT(x(\guard(y_j)), \guard(y_j)) \ge \dEDIT(x_i,y_j) + \Delta_L(-2\rho) + \Delta_R(2 \rho)$ as we show in the following. This implies the desired statement, since $\Delta_L(-2\rho) + \Delta_R(2 \rho) \ge \big||\guard(y_j)| - |x(\guard(y_j))|\big| \ge \gamma_4 - |x(\guard(y_j))| \ge \beta (\gamma_4 - |x(\guard(y_j))|)$. 
    To show the claim, we set $s_L := \min\{\Delta_L(k) \mid -2\rho \le k \le 0\}$ and $s_R := \min\{\Delta_R(k) \mid 0 \le k \le 2 \rho\}$. Note that $\Delta_L(k)$ has a total variation of at least $\Delta_L(-2\rho) - s_L + \Delta_L(0) - s_L$ over $k=-2\rho,\ldots,0$, since it starts in $\Delta_L(-2\rho)$, changes to $s_L$, and then changes to $\Delta_L(0)$. Thus, summing $|\Delta_L(k) - \Delta_R(k)| = |\Delta_L(k) - \Delta_L(k+1)|$ over all $-2\rho \le k \le -1$ yields at least $\Delta_L(-2\rho) - s_L + \Delta_L(0) - s_L$. Moreover, for every $-2\rho \le k < -1$ \ineq{xzk} applies and the summand $\csubst \cdot \min\{\Delta_L(k), \Delta_R(k)\}$ is at least $\csubst \cdot s_L$. As the number of such $k$'s is $2 \rho -1 \ge 2/\csubst$, 
    the total contribution of the summand $\csubst \cdot \min\{\Delta_L(k), \Delta_R(k)\}$ over all $k < 0$ is at least $2 s_L$. Thus, we have 
    \begin{align*}
      \sum_{k=-2\rho}^{-1} \dEDIT(x(z_k),z_k) &\ge \sum_{k=-2\rho}^{-1}  | \Delta_L(k) - \Delta_R(k) | + \sum_{k=-2\rho}^{-2} \csubst \cdot \min\{\Delta_L(k), \Delta_R(k)\}  \\
      &\ge \big(\Delta_L(-2\rho) - s_L + \Delta_L(0) - s_L\big) + \big(2 s_L \big) \ge \Delta_L(-2\rho) + \Delta_L(0).
    \end{align*}
    Using a symmetric statement for the sum over all $k>0$ as well as \eq{xzO}, we obtain the desired inequality $\dEDIT(x(\guard(y_j)),\guard(y_j)) = \sum_{k=-2\rho}^{2\rho} \dEDIT(x(z_k),z_k) \ge \dEDIT(x_i, y_j) + \Delta_L(-2\rho) + \Delta_R(2\rho)$.
  \end{proof}

  Since $L^\y = 0^{n\gamma_3}$ and $x(L^\y)$ is a prefix of $x$, by \claref{editprefix} we have $\dEDIT(x(L^\y),L^\y) \ge n\gamma_3 - \beta |x(L^\y)|$, and symmetrically we get $\dEDIT(x(R^\y),R^\y) \ge n\gamma_3 - \beta |x(L^\y)|$.
  By \facref{EDITgreedy}.(2), we have $\dEDIT(x(Z^\y_j),Z^\y_j) \ge \big| |Z^\y_j| - |x(Z^\y_j)| \big| \ge \beta ( \gamma_2 - |x(Z^\y_j)| )$.
  Putting all of this together, we obtain
  \begin{align*}
    \dEDIT(x,y) \ge 2n \gamma_3 + \beta \Big[ \sum_{j=1}^m ( \gamma_4 - |x(\guard(y_j))| ) + \sum_{j=1}^{m-1} ( \gamma_2 - |x(Z^\y_j)| ) - |x(L^\y)| - |x(R^\y)| \Big] +  \delta(A),
  \end{align*}
  where we used $\delta(A) = \sum_{(i,j) \in A} \dEDIT(x_i,y_j) + (m-|A|) \max_{i,j} \dEDIT(x_i,y_j)$.
  Note that by definition of $x$ and since the strings $x(\guard(y_j)), x(Z_j^\y), x(L^\y), x(R^\y)$ partition $x$ we have
  $$ n \gamma_4 + (n-1) \gamma_2 = |x| = \sum_{j=1}^m |x(\guard(y_j))| + \sum_{j=1}^{m-1} |x(Z^\y_j)| + |x(L^\y)| + |x(R^\y)|. $$
  Together, this yields the desired bound from below
  $$ \dEDIT(x,y) \ge 2 n \gamma_3 - \beta (n-m) (\gamma_4 + \gamma_2) + \delta(A). \qedhere $$
\end{proof}

\subsection{Algorithm}
\label{sec:EDITalgo}

For completeness, we prove a generalization of the algorithm of Hirschberg~\cite{hirschberg1977} from LCS to edit distance. Recall that the trivial dynamic programming algorithm computes a table storing all distances $\dEDIT(x[1..i],y[1..j])$. In contrast, we build a dynamic programming table storing for any index $j$ and any cost $k$ the minimal index $i$ with $\dEDIT(x[1..i],y[1..j]) - \cdelx(i-j) = k$. For some intuition, note that for $i \ge j$ at least $i-j$ symbols in $x[1..i]$ have to be deleted so that the cost $\dEDIT(x[1..i],y[1..j])$ is at least $\cdelx (i-j)$. Thus, it makes sense to ``normalize'' the cost by subtracting $\cdelx (i-j)$. As we will see, the normalized cost is bounded by $\Oh(m)$ (the length of the smaller of the two strings), which reduces the table size to $\Oh(m^2)$.

\begin{thm} \label{thm:editalgo}
  Let $\cdelx,\cdely,\cmatch,\csubst$ be positive integers.
  \EDITallc\ can be solved in time $\Oh((n + m^2) \log |\Sigma|)$ on strings of length $n,m$ with $n \ge m$ over alphabet $\Sigma$.
\end{thm}

Note that it is easy to ensure $\Sigma \subseteq [n+m]$ after $\Oh(n \log(\min\{|\Sigma|,n\}))$ preprocessing.\footnote{To compress the alphabet we build a balanced binary search tree $T$ whose nodes correspond to $\Sigma$ (by simply adding all symbols of $x$ and $y$ to $T$). Then we replace each symbol by its index in some fixed ordering of the nodes of $T$.} Thus, the running time is at most $\Oh((n+m^2) \log n) = \tilde \Oh(n+m^2)$, and \thmref{EDITalgo} follows from the above theorem and the second part of \lemref{editreduction}.
Our algorithm is designed for the pointer machine model; on the Word RAM the log-factor can be improved.

Consider strings $x,y$ over alphabet $\Sigma$ of length $n,m$, respectively, $n \ge m$. 
For convenience, we set $\min \emptyset := \infty$.
For any index $i \in \{0,\ldots,n\}$ and symbol $\sigma \in \Sigma$ we set
\begin{align*}
  \Next^x_{=\sigma}(i) &:= \min\{i' \mid i < i' \le n \text{ and } x[i'] = \sigma\},  \\
  \Next^x_{\ne\sigma}(i) &:= \min\{i' \mid i < i' \le n \text{ and } x[i'] \ne \sigma\}.
\end{align*}
We argue that a data structure can be built in $\Oh(n \log |\Sigma|)$ preprocessing time supporting $\Next^x_{=\sigma}(i)$ and $\Next^x_{\ne\sigma}(i)$ queries in time $\Oh(\log |\Sigma|)$. A simple solution with worse running time is to precompute all answers to all possible queries $\Next^x_{=\sigma}(i)$ and $\Next^x_{\ne\sigma}(i)$, with $i \in \{0,\ldots,n\}$, $\sigma \in \Sigma$, in time $\Oh(|\Sigma| n)$ by one scan from $x[n]$ to $x[1]$.
To improve the preprocessing time for $\Next^x_{\ne\sigma}(i)$, note that $\Next^x_{\ne\sigma}(i) = i+1$ for all $\sigma \ne x[i+1]$. Thus, we only have to precompute $\Next^x_{\ne x[i+1]}(i)$ (which can be done in time $\Oh(n)$ by one scan from $x[n]$ to $x[1]$), then $\Next^x_{\ne\sigma}(i)$ can be queried in time $\Oh(1)$.
For $\Next^x_{=\sigma}(i)$, for any $i \in \{0,\ldots,n\}$ we build a dictionary $D_i$ storing $\Next^x_{=\sigma}(i)$ for each $\sigma \in \Sigma$. Note that $D_{i-1}$ and $D_i$ differ only for the symbol $x[i+1]$. Thus, we can use persistent search trees~\cite{driscoll1986} as dictionary data structures, resulting in a preprocessing time of $\Oh(n \log |\Sigma|)$ for building $D_0,\ldots,D_n$ and a lookup time of $\Oh(\log |\Sigma|)$ for querying $\Next^x_{=\sigma}(i)$.
Using such a Next data structure, we can formulate our dynamic programming algorithm, see \algref{edit}. 

\begin{algorithm}
\caption{Algorithm for solving \EDITallc\ in time $\Oh((n + m^2) \log |\Sigma|)$.}
\label{alg:edit}
\begin{algorithmic}
\newcommand{\algrule}[1][.2pt]{\par\vskip.2\baselineskip\hrule height #1\par\vskip.3\baselineskip}
\renewcommand{\algorithmicrequire}{\textbf{Assumption:}}
  \Require $\cdelx,\cdely,\cmatch,\csubst$ are positive integers
\renewcommand{\algorithmicrequire}{\textbf{Input:}}
  \Require strings $x,y$ of length $n,m$, $n \ge m$
\renewcommand{\algorithmicensure}{\textbf{Output:}}
  \Ensure $\dEDIT(x,y)$
  \algrule
  \State $M \gets (\cdelx + \cdely) m$
  \State Implicitly set $I[j,k] \gets \infty$ for all $j$ and all $k < 0$ or $k > M$
  \State $I[0,0] \gets 0$
  \State $I[0,k] \gets \infty$ for $0 < k \le M$
  \For{$j=1,\ldots,m$}
    \For{$k=0,\ldots,M$}
      \begingroup
      \setlength{\abovedisplayskip}{0pt}
      \setlength{\belowdisplayskip}{0pt}
      \begin{flalign*} 
        \qquad \quad I[j,k] \gets \min\{ &I[j-1,k-\cdelx-\cdely], & \\
         &\Next_{=y[j]}^x(I[j-1,k-\cmatch]), & \\
         &\Next^x_{\ne y[j]}(I[j-1,k-\csubst])\} &
      \end{flalign*}
      \endgroup
    \EndFor
  \EndFor 
  \State \Return $\cdelx(n-m) + \min\{0 \le k \le M \mid I[m,k] < \infty\}$.
\end{algorithmic}
\end{algorithm}

Since $\Next^x_{=\sigma}$ and $\Next^x_{\ne\sigma}$ can be queried in time $\Oh(\log |\Sigma|)$ and $M = (\cdelx + \cdely) m = \Oh(m)$, \algref{edit} runs in time $\Oh(m^2 \log |\Sigma|)$. Together with the preprocessing time for the $\Next$ data structure, we obtain a total running time of $\Oh((n + m^2) \log |\Sigma|)$. It remains to argue correctness.

\paragraph{Correctness}
We prove that the dynamic programming table $I[j,k]$ has the following meaning.

\begin{lem}
  \algref{edit} computes for any $j \in [m]$, $k \in \mathbb{Z}$
$$ I[j,k] = \min\{ 0 \le i \le n \mid \dEDIT(x[1..i],y[1..j]) - \cdelx (i-j) = k\}. $$
\end{lem}
\begin{proof}
  Let $R[j,k] := \min\{ 0 \le i \le n \mid \dEDIT(x[1..i],y[1..j]) - \cdelx (i-j) = k\}$ be the right hand side of the statement.
  
  The statement is true for $j=0$, since for the empty string $\epsilon$ we have $\dEDIT(x[1..i],\epsilon) = \cdelx \cdot i$, so that $R[0,k] = 0$ for $k=0$, and $\infty$ otherwise, which is exactly how we initialize $I[0,k]$.
  
  We show that $R[j,k] = \infty$ for $k<0$ or $k > M$, which is also implicitly assumed for $I[j,k]$ in \algref{edit}. Note that for $i \ge j$ we have to delete at least $i-j$ symbols in $x[1..i]$ when traversing it with $y[1..j]$, which implies $\dEDIT(x[1..i],y[1..j]) \ge \cdelx (i-j)$. Since additionally for $i < j$ the term $- \cdelx (i-j)$ is positive, we have $\dEDIT(x[1..i],y[1..j]) - \cdelx (i-j) \ge 0$ for all $i,j$. Thus, for no $k<0$ we can have $\dEDIT(x[1..i],y[1..j]) - \cdelx (i-j) = k$, implying $R[j,k] = \infty$ in this case.
  Moreover, $\dEDIT(x[1..i],y[1..j]) \le \cdelx \cdot i + \cdely \cdot j$, which implies $\dEDIT(x[1..i],y[1..j]) - \cdelx (i-j) \le (\cdely + \cdelx) j \le M$. Thus, we also have $R[j,k] = \infty$ for $k > M$. 
  
  It remains to show the statement for $j > 1$ and $0 \le k \le M$. Inductively, we can assume that the statement holds for $j-1$. We show that $R[j,k]$ satisfies the same recursive equation as $I[j,k]$ in \algref{edit}.
  Let $i := R[j,k]$ and consider an optimal traversal $T$ of $(x[1..i],y[1..j])$. We obtain a traversal $T'$ by removing the last operation in $T$.
  
  If the last operation in $T$ is a deletion in $x$, then $T'$ is an optimal traversal of $(x[1..i-1],y[1..j])$ with cost $\dEDIT(x[1..i-1],y[1..j]) = \dEDIT(x[1..i],y[1..j])-\cdelx$. Thus, we can decrease $i$ to $i-1$ while keeping $k = \dEDIT(x[1..i],y[1..j]) - \cdelx(i-j) = \dEDIT(x[1..i-1],y[1..j]) - \cdelx(i-1-j)$. This contradicts minimality of $i = R[j,k]$, so the last operation in $T$ cannot be a deletion in $x$.
  
  If the last operation in $T$ is a deletion in $y$, then $T'$ is an optimal traversal of $(x[1..i],y[1..j-1])$ with cost $\dEDIT(x[1..i],y[1..j-1]) = \dEDIT(x[1..i],y[1..j])-\cdely$. Thus, we have
  \begin{align*}
    R[j,k] &= \min\{ 0 \le i \le n \mid \dEDIT(x[1..i],y[1..j]) - \cdelx (i-j) = k\}  \\
    &= \min\{0 \le i \le n \mid \dEDIT(x[1..i],y[1..j-1]) - \cdelx(i-(j-1)) = k - \cdely - \cdelx\}  \\
    &= R[j-1,k-\cdelx - \cdely].
  \end{align*}
  
  If the last operation in $T$ is a matching of $x[i]$ and $y[j]$, then $T'$ is an optimal traversal of $(x[1..i-1],y[1..j-1])$ with cost $\dEDIT(x[1..i-1],y[1..j-1]) = \dEDIT(x[1..i],y[1..j]) - \cmatch$. This implies $\dEDIT(x[1..i-1],y[1..j-1]) - \cdelx((i-1)-(j-1)) = k - \cmatch$, so that $i-1$ is a candidate for $R[j-1,k-\cmatch]$. 
  Let $i' := R[j-1,k-\cmatch]$ and note that $i' \le i-1$. 
  As $x[i] = y[j]$, we obtain $i \ge \Next^x_{=y[j]}(i') =: i^*$.
  In the following we show $i = i^*$.
  By definition of $i'$ we have $\dEDIT(x[1..i'],y[1..j-1]) - \cdelx (i'-j+1) = k - \cmatch$. Hence, $\dEDIT(x[1..i^*],y[1..j]) \le \cmatch + \cdelx (i^* - i' - 1) + \dEDIT(x[1..i'],y[1..j-1]) = k + \cdelx(i^* - j)$. We even have equality, since otherwise $\dEDIT(x[1..i],y[1..j]) \le \dEDIT(x[1..i^*],y[1..j]) + \cdelx(i-i^*) < k + \cdelx(i-j)$, contradicting the definition of $i$. Thus, $i^*$ is a candidate for $R[j,k]$, implying that we also have $i \le i^*$.
  Hence, we have $R[j,k] = i^* = \Next^x_{=y[j]}(i') = \Next^x_{=y[j]}(R[j-1,k-\cmatch])$.
  
  We argue analogously if the last operation in $T$ is a substitution of $x[i]$ and $y[j]$.
  This yields
  $$ R[j,k] = \min\{ R[j-1,k-\cdelx-\cdely], \Next_{=y[j]}^x(R[j-1,k-\cmatch]), \Next^x_{\ne y[j]}(R[j-1,k-\csubst])\}. $$
  Hence, $R[j,k]$ satisfies the same recursion as $I[j,k]$, and we proved $R[j,k] = I[j,k]$ for all $j,k$. 
\end{proof}

\begin{lem}
  \algref{edit} correctly computes $\dEDIT(x,y)$.
\end{lem}
\begin{proof}
  Among all optimal traversals of $(x,y)$, pick a traversal $T$ that ends with the maximal number~$d$ of deletions in $x$, and set $i := n-d$. 
  Observe that $i$ is minimal with $\dEDIT(x[1..i],y[1..m]) + \cdelx (n-i) = \dEDIT(x,y)$, which is equivalent to $\dEDIT(x[1..i],y[1..m]) - \cdelx(i-m) = \dEDIT(x,y) - \cdelx(n-m) =: k$. 
  Thus, $i = I[m,k] < \infty$, which implies that the return value of \algref{edit} is at most $\cdelx(n-m) + k = \dEDIT(x,y)$.
  
  Moreover, for any $k$ with $I[m,k] < \infty$ there is a $0 \le i \le n$ with $\dEDIT(x[1..i],y[1..m]) - \cdelx(i-m) = k$. By appending $n-i$ deletions in $x$ to any optimal traversal of $(x[1..i],y[1..m])$, we obtain $\dEDIT(x,y) \le \dEDIT(x[1..i],y[1..m]) + \cdelx(n-i) = k + \cdelx(n-m)$. Hence, the return value of \algref{edit} is also at least $\dEDIT(x,y)$.
\end{proof}

\section{Dynamic Time Warping} \label{sec:dtw}

We present coordinate values and an \emph{unbalanced} \aligngad\ gadget for \DTW\ on one-dimensional curves taking values in $\mathbb{N}_0$, i.e., we consider the set of inputs ${\cal I} := \bigcup_{k \ge 0} \mathbb{N}_0^k$. 

\begin{lem}
\DTW\ admits coordinate values by setting
\begin{align*}
\oleft := 1100,\;  \zleft := 0110,\;  \oright := 0011,\;  \zright := 1010.
\end{align*}
\end{lem}
\begin{proof}
All four values have the same length and sum of all entries, so they have equal type. Short calculations show that $4=\dDTW(\oleft,\oright)>\dDTW(\zleft,\oright)=\dDTW(\zleft,\zright)=\dDTW(\oleft,\zright)=1$.
\end{proof}

\begin{defn}
\label{def:C&Adtw}
Consider instances $x_1,\ldots,x_n \in \inputs_{t_\x}$ and $y_1,\ldots,y_m \in \inputs_{t_\y}$ with $n \ge m$ and types $t_\x = (\ell_\x,s_\x), t_\y = (\ell_\y,s_\y)$. We
define $M:=2z$, where $z$ is the largest value contained in any
of the one-dimensional curves $x_{1},\dots,x_{n},y_{1},\dots,y_{m}$, and we set $\kappa:=3(\ell_{\x}+\ell_{\y})$.
We construct 
\begin{eqnarray*}
\CandA_{\x}^{m,t_{\y}}(x_{1},\dots,x_{n}) & := & M^{\kappa} \; x_{1} \; M^{\kappa} \; x_{2} \; M^{\kappa} \,\dots\, M^{\kappa} \; x_{n} \; M^{\kappa},\\
\CandA_{\y}^{n,t_{\x}}(y_{1},\dots,y_{m}) & := & M^{\kappa} \; y_{1} \; M^{\kappa} \; y_{2} \; M^{\kappa} \,\dots\, M^{\kappa} \; y_{m} \; M^{\kappa},
\end{eqnarray*}
where $M^\kappa$ is to be understood as a sequence with $\kappa$ times the entry $M$.
\end{defn}
\begin{lem}\label{lem:C&Adtw}
Definition~\ref{def:C&Adtw} realizes an unbalanced \aligngad\
gadget for dynamic time warping.\end{lem}

Thus, \thmref{main} is applicable, implying a lower bound of $\Oh((n m)^{1-\eps})$ for \DTW\ on one-dimensional curves over $\mathbb{N}_0$. 
To restrict the alphabet further, note that our basic values use the alphabet $\{0,1\} \subseteq \mathbb{N}_0$ and each invocation of the \aligngad\ gadget introduces a new symbol which is twice as large as the largest value seen so far. Since in the proof of \thmref{main} we use \aligngad\ gadgets three times, we introduce the symbols 2, 4, and 8. In total, we prove quadratic-time hardness of \DTW\ on one-dimensional curves taking values in $\{0,1,2,4,8\} \subseteq \mathbb{N}_0$.
This proves \thmrefs{dtw}{unbalanced}. 

\begin{proof}[Proof of Lemma~\ref{lem:C&Adtw}]

Observe that $x:=\CandA_{\x}^{m,t_{\y}}(x_{1},\dots,x_{n})$ and
$y:=\CandA_{\y}^{n,t_{\x}}(y_{1},\dots,y_{m})$ can be computed in time $\bigOh((n+m)(\ell_{\x}+\ell_{\y}))$
yielding strings of length $\bigOh(n(\ell_{\x}+\ell_{\y}))$ and $\bigOh(m(\ell_{\x}+\ell_{\y}))$,
respectively. Moreover, $\type(x)$ and $\type(y)$ only depend on $t_\x,t_\y,n,m$. 
It remains to show the inequalities \eqref{eq:Cone} of \defref{cagadget}, for which we set $C := (n-m)(\ell_\x M - s_\x)$.

We start with the following useful observations.
\begin{claim} \label{cla:usefuldtw}
Let $\ell \ge 1 $ and $a,a',b,b'\in \mathbb{N}_0$. For any $i \in [n], j \in [m]$, we have
\begin{enumerate}
\item[(1)] $\dDTW(x_i, M^{\ell})\ge\dDTW(x_i, M) = \ell_{\x}M-s_{\x} \ge \ell_\x M /2$ and $\dDTW(M^{\ell},y_{j})\ge\dDTW(M,y_{j})=\ell_{\y}M-s_{\y} \ge \ell_\y M /2$,
\item[(2)] $\dDTW(x_i,y_j) < (\ell_\x + \ell_\y)M/2$.
\item[(3)] $\dDTW(x',M^{\kappa}) \ge \kappa M/2$ and $\dDTW(M^{\kappa},y') \ge \kappa M/2$ for any substrings $x'$ of $x_i$ and $y'$ of $y_j$,
\item[(4)] $\dDTW(M^{a}x_{i}M^{a'},M^{b}y_{j}M^{b'}) \ge \dDTW(x_{i},y_{j})$.
\end{enumerate}
\end{claim}
\begin{proof}
For~(1), observe that each symbol of $x_{i}$ can only be traversed together with the symbol $M$ and hence,
\[\dDTW(x_{i},M^{\ell}) \ge\dDTW(x_{i},M)=\sum_{k=1}^{\ell_{\x}}|M-x_{i}[k]|=\ell_{\x}M-\sum_{k=1}^{\ell_{\x}}x_{i}[k]=\ell_{\x}M-s_{\x}.\]
Since $x_i[k] \le z = M/2$, we have $s_\x \le \ell_\x M/2$. The statement for $y_j$ is symmetric.

For~(2) and~(3), note that all symbols in $x'$ are in $[0,z]$. Hence, we obtain $\dDTW(x_i,y_j) \le \max\{|x_i|,|y_j|\}\cdot z < (\ell_\x + \ell_\y)M/2$. Likewise, $\dDTW(x',M^{\kappa}) \ge \kappa(M-z)=\kappa M/2$. The inequality for $y_j$ follows symmetrically.

To prove~(4), consider an optimal traversal $T$ of $M^{a}x_iM^{a'}$ and $M^{b}y_jM^{b'}$. We construct a traversal $T'$ of $x_i$ and $y_j$ that has no larger cost. If $T$ does not already traverse $x_i[1]$ together with $y_j[1]$, then at some step in $T$ either a symbol in $x_i$ is traversed together with a symbol of the prefix $M^b$ or a symbol in $y_j$ is traversed together with a symbol of the prefix $M^a$. Let us assume the first case, since the second is symmetric. A contiguous part $T^H$ of $T$ consists of traversing a prefix $x'$ of $x_i$ together with all symbols in $M^b$, incurring a cost of at least $|x'|M/2$. Let $T^R$ be the remaining part of $T$ after $T^H$. We construct a traversal $T''$ of $x_iM^{a'}$ and $y_jM^{b'}$ as follows. We first traverse $x'$ together with $y_j[1]$ and then follow $T^R$, which is possible since $T^R$ starts at $y_j[1]$. Since traversing $x'$ together with $y_j[1]$ incurs a cost of at most $|x'|z= |x'|M/2$, which is smaller than the cost of~$T^H$, the cost of our constructed traversal $T''$ is no larger than the cost of $T$. Symmetrically, we eliminate the suffixes $M^{a'}$ and $M^{b'}$ and construct a traversal $T'$ of $x_i$ and $y_j$ of cost no larger than $T$.
\end{proof}

We first verify that
\[
\dDTW(x,y)\le(n-m)(\ell_{\x}M-s_{\x})+\min_{A\in\strc_{n,m}}\sum_{(i,j)\in A}\dDTW(x_{i},y_{j}),
\]
by designing a traversal (illustrated in \figref{dtw}) that achieves this bound. Let $A\in\strc_{n,m}$
be the alignment minimizing the expression, and note that $A = \{(\Delta+1,1),\dots,(\Delta+m,m)\}$ for some $0 \le \Delta \le n-m$. We first traverse
$M^{\kappa}x_{1}M^{\kappa}\dots M^{\kappa}x_{\Delta}$ together with the first symbol
of $y$, $M$, which contributes a cost of $\sum_{i=1}^{\Delta}\dDTW(x_{i},M)=\Delta(\ell_{\x}M-s_{\x})$.
For $i=1,\dots,m$ we repeat the following: We traverse $M^{\kappa}x_{\Delta+i}$ together
with $M^{\kappa}y_{i}$ by traversing $M^\kappa$ and $M^\kappa$ simultaneously, and $x_i$ and $y_i$ in a locally optimal manner; this incurs a cost of $\dDTW(x_{\Delta+i},y_{i})$
for each $i$. Finally, we traverse the last block $M^\kappa$ in $y$ with the current block $M^\kappa$ in $x$, and then traverse the remainder $x_{\Delta+m+1}M^{\kappa}\dots M^{\kappa}x_{n}M^{\kappa}$ of $x$
together with the last symbol of $y$, $M$. The total cost amounts
to $\Delta(\ell_{\x}M-s_{\x})+\sum_{i=1}^{m}\dDTW(x_{\Delta+i},y_{i})+(n-\Delta-m)(\ell_{\x}M-s_{\x})=(n-m)(\ell_{\x}M-s_{\x})+\sum_{(i,j)\in A}\dDTW(x_{i},y_{j})$.

\begin{figure}
\includegraphics[width=\textwidth]{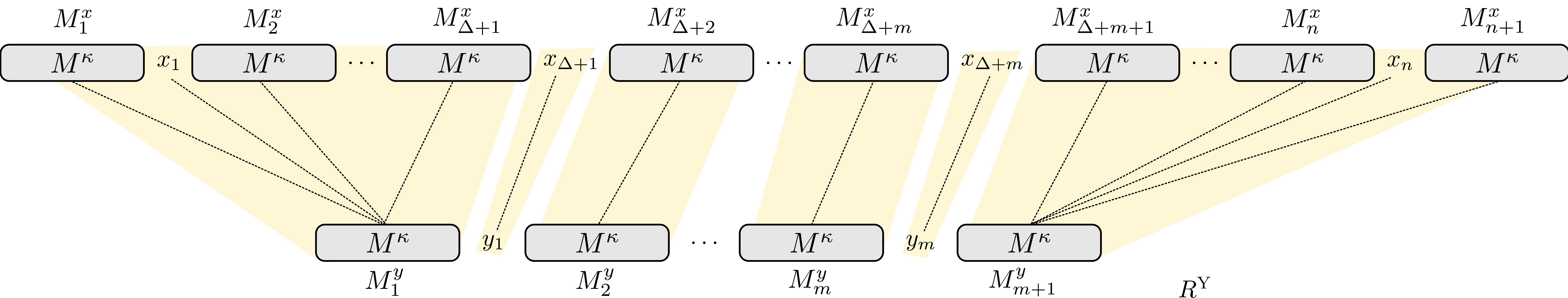}
\caption{Optimal traversal corresponding to structured alignment $A=\{(\Delta+j,j) \mid j\in [m]\} \in \strc_{n,m}$.}
\label{fig:dtw}
\end{figure}

In the remainder of the proof, we verify that
\begin{align*}
\dDTW(x,y)\ge(n-m)(\ell_{\x}M-s_{\x})+\min_{A\in\algn_{n,m}} \Big[ \sum_{(i,j)\in A}\dDTW(x_{i},y_{j})+(m-|A|)\max_{i,j}\dDTW(x_{i},y_{j}) \Big].
\end{align*}
Let $T^* = ((a_1^{*},b_1^*),\ldots,(a_t^{*},b_t^{*}))$ be an optimal traversal of $(x,y)$ (see \secref{preliminaries} for the definition of traversals). 
Substrings $x'$ of $x$ and $y'$ of $y$ are \emph{paired}
if for some index $i$ in $x'$ and some index $j$
in $y'$ we have $(i,j) = (a_{t'}^*,b_{t'}^*)$ for some $1 \le t' \le t$.

We call the $i$-th occurrence of $M^\kappa$ in $x$ the \emph{$i$-th $M$-block} $M^x_i$ of $x$, and similarly for $y$. Let $X := \{M^x_i \mid i \in [n+1]\}$, $Y := \{M^y_j \mid j \in [m+1]\}$ be the sets of all $M$-blocks of $x$ and $y$, respectively.
We define a bipartite graph $G_M$ with vertex set $X \cup Y$, where $M$-blocks $M^x_i$ and $M^y_j$ are connected by an edge if and only if they are paired. We show the following properties of $G_M$.

\begin{claim}[Planarity]
  For any paired $M^x_i,M^y_j$ and paired $M^x_{i'},M^y_{j'}$ we have $i\le i'$ and $j \le j'$ (or $i \ge i'$ and $j \ge j'$). 
\end{claim}
\begin{proof}
  By monotonicity of traversals, for $k \le k'$ we have $a_k^* \le a_{k'}^*$ and $b_k^* \le b_{k'}^*$. Thus, if $x[a_k^*]$ is in $M^x_i$ and $x[a_{k'}^*]$ is in $M^x_{i'}$, then $i \le i'$. Similarly, if $y[b_k^*]$ is in $M^y_j$ and $y[b_{k'}^*]$ is in $M^y_{j'}$, then $j \le j'$. Hence, for any paired $M^x_i,M^y_j$ and $M^x_{i'},M^y_{j'}$ we have $i\le i'$, $j \le j'$ or $i \ge i'$, $j \ge j'$. 
\end{proof}

\begin{claim}
  $G_M$ has no isolated vertices.
\end{claim}
\begin{proof}
Assume that some $M$-block $M_{i}^{x}$ is not paired with any $M$-block of $y$, and let $i$ be maximal with this property. Note that $i < n+1$, as the last $M$-block of $x$ is always paired with the last $M$-block of $y$.
Then there is some $j \in [m]$ such that $M^x_i$ is paired with $y_j$, but $M^x_i$ is not paired with any part of $y$ outside $y_j$. By maximality of $i$ and planarity, $M^y_{j+1}$ is paired with $x_{i}$ or $M^x_{i+1}$, as otherwise $M^x_{i+1}$ is not paired with any $M^y_{j'}$.

We can find a cheaper traversal as follows. Consider the first time $t_1$ at which the traversal $T^*$ is simultaneously at the first symbol of $M^x_i$ and any symbol of $y_j$ (this exists since $M^x_i$ is paired to $y_j$, but to no part of $y$ outside $y_j$), and any time $t_2$ at which $T^*$ is at $M^y_{j+1}$ and $x_{i}$ or at $M^y_{j+1}$ and $M^x_{i+1}$. Between $t_1$ and $t_2$, $T^*$ has a cost of at least $\dDTW(y',M^\kappa)$, where $y'$ is any substring of $y_j$. By \claref{usefuldtw}.(3), this is at least $\kappa M/2$. 
We replace this part of $T^*$ by traversing (i) the remainder of $y_j$ with the first symbol of $M^x_i$, (ii) $M^x_i$ with the necessary part of $M^y_{j+1}$, and (iii) the necessary part of $x_i$ and $M^x_{i+1}$ with the current symbol in $y$, $M$. By \claref{usefuldtw}.(1), this incurs a cost of at most $\dDTW(x_i,M) + \dDTW(y_j,M) = \ell_\x M - s_\x + \ell_\y M - s_\y \le (\ell_\x + \ell_\y)M$. By our choice of $\kappa = 3(\ell_{\x}+\ell_{\y})$, we improve the cost of the traversal, contradicting optimality of~$T^*$. This shows that no vertex in $X$ is isolated, we argue similarly for vertices in $Y$.
\end{proof}

\begin{claim}
  $G_M$ contains no path of length 3.
\end{claim}
\begin{proof}
Assume that $G_M$ contains a path $M^x_i - M^y_j - M^x_{i'} - M^y_{j'}$. Without loss of generality we assume $i < i'$, the case $i > i'$ is symmetric. By planarity, we have $j < j'$. Since $G_M$ has no isolated vertices and by planarity, every $M^x_{i''}$ with $i \le i'' \le i'$ is paired with $M^y_j$, so we can assume that $i' = i+1$ (after replacing $i$ with $i'-1$). Similarly, we can assume $j' = j+1$, and the path is $M^x_i - M^y_j - M^x_{i+1} - M^y_{j+1}$.

We can find a cheaper traversal as follows. 
Consider any time $t_1$ at which the traversal $T^*$ is simultaneously at $M^x_i$ and $M^y_j$ (this exists since $M^x_i$ and $M^y_j$ are paired), and consider any time $t_2$ at which $T^*$ is simultaneously at $M^x_{i+1}$ and $M^y_{j+1}$. Between $t_1$ and $t_2$, $T^*$ traverses $x_i$ with (parts of) $M^y_j$, and $y_j$ with (parts of) $M^x_{i+1}$, which by \claref{usefuldtw}.(1) incurs a cost of at least $\dDTW(x_i,M) + \dDTW(M,y_j) \ge (\ell_\x + \ell_\y) M/2$. We replace this part of $T^*$ by traversing (i) the remaining parts of $M^x_i$ and $M^y_j$, (ii) $x_i$ and $y_j$ (in a locally optimal way), and (iii) the necessary parts of $M^x_{i+1}$ and $M^y_{j+1}$. This incurs a cost of $\dDTW(x_i,y_j) < (\ell_\x + \ell_\y) M/2$ (by \claref{usefuldtw}.(4)), which contradicts optimality of $T^*$. 
\end{proof}

By the above two claims, $G_M$ is a disjoint union of stars. By planarity and since $G_M$ has no isolated vertices, the leafs of any star in $G_M$ have to be consecutive $M$-blocks. Hence, we can write the components of $G_M$ as $C_1,\ldots,C_s$ with $C_k = \{M^x_{i_k}\} \cup \{M^y_{j_k}, M^y_{j_k+1}, \ldots, M^y_{j_k+d_k-1}\}$, and $C'_1,\ldots,C'_{s'}$ with $C'_k = \{M^y_{j'_k}\} \cup \{M^x_{i'_k}, M^x_{i'_k+1}, \ldots, M^x_{i'_k+d'_k-1}\}$.

\begin{claim} \label{cla:compsizes}
  We have $\sum_{k=1}^s d_k = m-s'+1$ and $ \sum_{k=1}^{s'} d'_k = n-s+1$.
\end{claim}
\begin{proof}
  Since the components $C_1,\ldots,C_s$ and $C'_1,\ldots,C'_{s'}$ partition $G_M$, restricted to $Y$ we have $s' + \sum_{k=1}^s d_k = \sum_{k=1}^{s'} |C'_k \cap Y| + \sum_{k=1}^s |C_k \cap Y| = |Y| = m+1$. The second claim follows analogously.
\end{proof}

We construct an alignment by aligning the $x_i,y_j$ that lie between two consecutive components of $G_M$. More formally, we define an alignment $A$ by aligning $(i_k-1,j_k-1)$ (for all $k \in [s]$ with $i_k, j_k > 1$) and aligning $(i'_k-1,j'_k-1)$ (for all $k \in [s']$ with $i'_k,j'_k > 1$). Since $G_M$ has no isolated vertices, $A$ is a valid alignment. We have $|A| = s+s'-1$, since only the leftmost component of $G_M$ has $i_k=1$, $j_k=1$, $i'_k=1$, or $j'_k=1$, and all other components give rise to exactly one aligned pair.

Let us calculate the cost of $T^*$. Each $y_j$ that lies between the leafs of a star $C_k$ in $G_M$ (i.e., $j_k \le j < j_k + d_k$) has to be traversed together with (parts of) $M^x_{i_k}$. By \claref{usefuldtw}.(1), this incurs a cost of at least $\dDTW(M,y_j) = \ell_\y M - s_\y$. Likewise, each $x_i$ that lies between the leafs of a star $C'_k$ incurs a cost of at least $\ell_\x M - s_\x$.
For any $(i,j) \in A$, $x_i$ is traversed together with a substring of $M^\kappa y_j M^\kappa$, and $y_j$ is traversed together with a substring of $M^\kappa x_i M^\kappa$. Hence, there are $a,a',b,b' \ge 0$ such that we traverse $M^a x_i M^{a'}$ together with $M^b y_j M^{b'}$. 
By \claref{usefuldtw}.(4), this incurs a cost of at least $\dDTW(x_i,y_j)$. 
In total, the cost of the optimal traversal $T^*$ is 
$$ \dDTW(x,y) \ge \sum_{k=1}^s (d_k-1) (\ell_\y M - s_\y) + \sum_{k=1}^{s'} (d'_k - 1) (\ell_\x M - s_\x) + \sum_{(i,j) \in A} \dDTW(x_i,y_j). $$
By \claref{compsizes}, we have $\sum_{k=1}^s (d_k-1) = m - (s+s'-1) = m - |A|$. Similarly, $\sum_{k=1}^{s'} (d'_k-1) = n-|A| = (n-m)+(m-|A|)$. Additionally bounding $\ell_\y M -s_\y + \ell_\x M - s_\x \ge (\ell_\x + \ell_\y)M/2 > \max_{i,j} \dDTW(x_i,y_j)$, we obtain the desired inequality
$$ \dDTW(x,y) \ge (m-|A|) \max_{i,j} \dDTW(x_i,y_j) + (n-m) (\ell_\x M - s_\x) + \sum_{(i,j) \in A} \dDTW(x_i,y_j). \qedhere $$
\end{proof}

\section{Palindromic and Tandem Subsequences}
\label{sec:LPSLTS}

In this section, we prove quadratic-time hardness of longest palindromic subsequence (LPS) and longest tandem subsequence (LTS) by presenting reductions from LCS. This proves \thmref{LPSLTS}. We will use the following simple facts about LCS, where we regard LCS as a minimization problem by defining $\dLCS(x,y) := |x| + |y| - 2|\LCS(x,y)|$. In the whole section we let $\Sigma$ be any alphabet with $0,1 \in \Sigma$. 

\begin{fact} \label{fac:lcsgreedyStrong}
    Let $z,w$ be binary strings and $\ell,k \in \mathbb{N}_0$. Then we have (1) $\dLCS(1^k z, 1^k w) = \dLCS(z,w)$, (2) $\dLCS(1^k z, w) \ge \dLCS(z,w)-k$ and (3) $\dLCS(0^\ell z, 1^k w) \ge \min\{k, \dLCS(z,1^k w) + \ell \}$. We obtain symmetric statements by flipping all bits and by reversing all involved strings.
\end{fact}
\begin{proof}
    (1) is a restatement of \claref{lcsgreedy}.(1). (2) follows from \facref{EDITgreedy}.(2). For (3), fix a \LCS\ $s$ of $(0^\ell z, 1^k w)$. If $s$ starts with a 0, then it does not contain the leading $1^k$ of the second argument, leaving at least $k$ symbols unmatched, so that $\dLCS(0^\ell z, 1^k w) \ge k$. Otherwise, if $s$ starts with a 1, then it does not contain the leading $0^\ell$ of the first argument, so that $|\LCS(0^\ell z, 1^k w)| = |\LCS(z, 1^k w)|$. Then we have $\dLCS(0^\ell z, 1^k w) = |0^\ell z| + |1^k w| - 2|\LCS(0^\ell z, 1^k w)| = \ell + |z| + |1^k w| - 2|\LCS(z, 1^k w)| = \ell + \dLCS(z,1^k w)$.
\end{proof}

\subsection{Longest Palindromic Subsequence}

We show that computing the length of the longest palindromic subsequence is essentially computationally equivalent to computing the length of the longest common subsequence of two strings. For completeness, we provide the following well known result which shows that LPS can be reduced to LCS in linear time. Recall that for a string $x$ we denote the reversed string by $\rev(x)$.

\begin{fact}[Folklore]
For any input $x\in \Sigma^*$, we have $|\LPS(x)| = |\LCS(x,\rev(x))|$.
\end{fact}
\begin{proof}
Let $p$ be a palindromic subsequence of $x$. Then $p = \rev(p)$ is a common subsequence of $x$ and $\rev(x)$, yielding $|\LCS(x,\rev(x))| \ge |\LPS(x)|$.

For the other direction, let $c$ be any LCS of $x$ and $\rev(x)$ of length $\ell$. It remains to show that we can find a palindromic subsequence $p$ of $x$ with $|p| \ge \ell$ (observe that $c$ itself is not necessarily a palindrome). Note that $c$ gives rise to a sequence of pairs $(a_1, b_1),\dots, (a_\ell, b_\ell)$ such that $a_1 < \cdots < a_\ell$, $b_1 > \cdots > b_\ell$, and $c = (x[a_1],\dots,x[a_\ell]) = (x[b_1],\dots,x[b_\ell])$. Define $m:= \lfloor \frac{\ell}{2} \rfloor + 1$. If $a_m \le b_m$, then $a_1 < \cdots < a_m \le b_m < \cdots < b_1$ and hence 
$(x[a_1],\dots,x[a_{m-1}],x[a_m],x[b_{m-1}],\dots,x[b_1])$ 
is a palindromic subsequence of $x$ of length 
$2m - 1 = 2 \lfloor \frac{\ell}{2} \rfloor + 1 \ge \ell$. Otherwise, i.e., if $a_m > b_m$, then $b_\ell < \cdots < b_m < a_m < \cdots < a_\ell$ gives rise to the palindromic subsequence $(x[b_\ell], \dots, x[b_m], x[a_m], \dots, x[a_\ell])$ of $x$ with length $2(\ell - m + 1) = 2\ell - 2\lfloor \frac{\ell}{2} \rfloor \ge \ell$.
\end{proof}

To prove our lower bound for computing a longest palindromic subsequence of a string $x$, we present a simple reduction from LCS to LPS, and then appeal to our lower bound for \LCS, which is equivalent to $\Edit(1,1,0,2)$, see \thmref{EDIT}.

\begin{thm}\label{thm:lps}
On input $x,y\in \Sigma^*$, we can compute, in time $\bigOh(|x|+|y|)$, a string $z\in \Sigma^*$ and $\kappa \in \mathbb{N}$ such that $|\LPS(z)| = 3\kappa + 2|\LCS(x,y)|$. 
\end{thm}
\begin{proof}
Let $\kappa := 2(\ell_\x+\ell_\y+1)$, where $\ell_\x := |x|$, $\ell_\y := |y|$. 
We define 
\[ z:= \;\; x \; 0^\kappa \; 1^\kappa \; 0^\kappa \; \rev(y). \]
Clearly, $z$ and $\kappa$ can be computed in time $\bigOh(\ell_\x + \ell_\y)$. 
Let $s$ be a LCS of $x$ and $y$. Then $s 0^\kappa 1^\kappa 0^\kappa \rev(s)$ is a palindromic subsequence of $z$, which proves $|\LPS(z)| \ge 3\kappa + 2|\LCS(x,y)|$.

To show $|\LPS(z)| \le 3\kappa + 2|\LCS(x,y)|$, fix a LPS $p$ of $z$ and let $\ell$ be its length. We define $m:=\lfloor \frac{\ell}{2}\rfloor$ and denote by $p_1=(p[1],\dots,p[m])$ the first ``half'' of $p$. Let $z_1=(z[1],\dots,z[i])$ be the shortest prefix of $z$ that contains $p_1$ as a subsequence and let $z_2 := (z[i+1],\dots,z[|z|])$ be the remainder of $z$. Then $p_1$, which by definition equals $(p[\ell], \dots, p[\ell-m+1])$, is a subsequence of $\rev(z_2)$. This shows that if $\ell$ is even, then $\ell \le 2|\LCS(z_1, \rev(z_2))|$. If $\ell$ is odd, we may without loss of generality assume that $p[m+1] = z_2[1]$. Hence $\rev(p_1)$ is a subsequence of $z_2' := (z_2[2], \dots,z_2[|z_2|])$, so that $\ell \le 2|\LCS(z_1,\rev(z_2'))|+1$. It remains to show that (i) $|\LCS(z_1,\rev(z_2))|\le \frac{3}{2}\kappa + |\LCS(x,y)|$ and (ii) $|\LCS(z_1,\rev(z_2'))|\le \frac{3}{2}\kappa + |\LCS(x,y)| - \frac{1}{2}$.

Assume that $|z_1| \le \ell_\x + \kappa$ or $|z_2| \le (\ell_\y+1) + \kappa$, then by $|\LCS(x,y)| \le \min\{|x|,|y|\}$ we obtain that $|\LCS(z_1,\rev(z_2'))| \le |\LCS(z_1,\rev(z_2))| \le \max\{\ell_\x,\ell_\y+1\} + \kappa < \frac{3}{2}\kappa + |\LCS(x,y)|$. Hence without loss of generality, $z_1 = x 0^\kappa 1^a$ and $z_2 = 1^{a'} 0^\kappa \rev(y)$ with $a'\ge 1$, where we assume that $a'\ge a$ since the other case is symmetric. Note that (i) and (ii) are equivalent to $\dLCS(z_1,\rev(z_2))\ge \dLCS(x,y)$ and $\dLCS(z_1,\rev(z_2')) \ge \dLCS(x,y)$, respectively. We compute
\begin{align*}
\dLCS(z_1,\rev(z_2)) & = \dLCS(x \, 0^\kappa\, 1^a, y \,0^\kappa \, 1^{a'}) & & \\
& = \dLCS(x \, 0^\kappa , y \, 0^\kappa \, 1^{a'-a}) & &\text{(by \facref{lcsgreedyStrong}.(1))} \\
& \ge \min\{\kappa, \dLCS(x \, 0^\kappa , y \, 0^\kappa )\} & &\text{(by \facref{lcsgreedyStrong}.(3))} \\
& = \min\{\kappa, \dLCS(x, y)\} = \dLCS(x,y). & &\text{(by \facref{lcsgreedyStrong}.(1))}.
\end{align*}
By replacing $a'$ by $a'-1\ge 0$, we obtain $\dLCS(z_1,\rev(z_2')) \ge \dLCS(x,y)$ by the same calculation. This yields $|\LPS(z)| = \ell \le 3 \kappa + 2|\LCS(x,y)|$, as desired.

\end{proof}

\subsection{Longest Tandem Subsequence}

As for LPS, our lower bound for LTS follows from a simple reduction from LCS and appealing to our lower bound for $\LCS$ of Theorem~\ref{thm:EDIT}.

\begin{thm}\label{thm:lts}
On input $x,y \in \Sigma^*$, we can compute, in time $\bigOh(|x|+|y|)$, a string $z \in \Sigma^*$ and $\kappa \in \mathbb{N}$ such that $|\LTS(z)| = 4\kappa + 2|\LCS(x,y)|$.
\end{thm}
\begin{proof}
Let $\kappa := \ell_\x+\ell_\y$, where $\ell_\x := |x|$ and $\ell_\y := |y|$. We define 
\[ z:= \;\; 0^\kappa \; x \; 1^\kappa \; 0^\kappa \;  y \; 1^\kappa . \]
Clearly, $z$ can be computed in time $\bigOh(\ell_\x + \ell_\y)$. Let $s$ be a LCS of $x$ and $y$. Then $t:= t'\, t'$ with $t':=0^\kappa s 1^\kappa$ is a tandem subsequence of $z$. Hence, we have $|\LTS(z)| \ge |t| = 4\kappa + 2|\LCS(x,y)|$.

To show $|\LTS(z)| \le 4\kappa + 2|\LCS(x,y)|$, fix a LTS $t=t' \,t'$ of $z$. Let $i$ be the smallest index such that $t'$ is a subsequence of $z_1 := (z[1],\dots,z[i])$ and let $z_2 := (z[i+1],\dots,z[|z|])$. By choice of $t$, $t'$ is also a subsequence of $z_2$, so that $|\LTS(z)| = 2|t'| \le 2 |\LCS(z_1, z_2)|$. Thus, it remains to prove that $2|\LCS(z_1,z_2)| \le 4\kappa + 2|\LCS(x,y)|$.

Assume that $|z_1| \le \kappa + \ell_\x$ or $|z_2| \le \kappa + \ell_\y$. Then, using $|\LCS(x,y)| \le \min\{|x|,|y|\}$, we conclude that $2|\LCS(z_1,z_2)| \le 2\kappa + 2(\ell_\x+\ell_\y) \le 4\kappa + 2|\LCS(x,y)|$. 

Hence, without loss of generality, we have (i) $z_1 = 0^\kappa x 1^\ell$ and $z_2 = 1^{\ell'} 0^\kappa y 1^\kappa$ or (ii) $z_1 = 0^\kappa x 1^\kappa 0^\ell$ and $z_2 = 0^{\ell'} y 1^\kappa$, for some $\ell, \ell'$ with $\ell + \ell' = \kappa$. We only consider case (i), since case (ii) is symmetric.
Note that $2|\LCS(z_1,z_2)| \le 4\kappa + 2|\LCS(x,y)|$ is equivalent to $\dLCS(z_1,z_2) \ge \dLCS(x,y)$. We obtain
\begin{align*}
\dLCS(z_1,z_2) & = \dLCS(0^\kappa x 1^\ell, 1^{\ell'} 0^\kappa y 1^\kappa) & & \\
& \ge \min\{ \kappa, \dLCS(0^\kappa x 1^\ell, 0^\kappa y 1^\kappa) + \ell' \} & & \text{(by \facref{lcsgreedyStrong}.(3))} \\
& = \min\{ \kappa, \dLCS(x 1^\ell, y 1^\kappa) + \ell' \} & & \text{(by \facref{lcsgreedyStrong}.(1))} \\
& = \min\{ \kappa, \dLCS(x , y 1^{\kappa-\ell}) + \ell' \} & & \text{(by \facref{lcsgreedyStrong}.(1))} \\
& \ge \min\{ \kappa, \dLCS(x , y) - (\kappa - \ell) + \ell' \} & & \text{(by \facref{lcsgreedyStrong}.(2))} \\
& = \min\{ \kappa, \dLCS(x , y) \} = \dLCS(x,y), 
\end{align*}
which proves the desired inequality $2|\LCS(z_1,z_2)| \le 4\kappa + 2|\LCS(x,y)|$.
\end{proof}

\section{Conclusion}

We prove conditional lower bounds for natural polynomial-time problems: Edit distance for general operation costs, including its special case longest common subsequence, dynamic time warping, longest palindromic subsequence, and longest tandem subsequence. Our results give strong evidence that the known algorithms for these problems are optimal up to lower order factors, even restricted to binary strings and one-dimensional curves, respectively.
We hope that the underlying framework will find application in hardness proofs for further similarity measures, and that the studied problems serve as starting points for further reductions. 

It remains an open question whether constant-factor approximations running in strongly subquadratic time can be ruled out for the above problems assuming SETH. Furthermore, most polynomial-time lower bounds show quadratic-time barriers, and it is challenging to prove matching SETH-based lower bounds for problems with, say, cubic or $\bigOh(n^{3/2})$-time algorithms (only few results are known in this direction~\cite{abboud_quadratic_2015,PatrascuW10}).

\bibliographystyle{plain}
\bibliography{edit}

\end{document}